\documentclass[sigconf, nonacm]{acmart}

\usepackage{graphicx} 
\usepackage{hyperref}
\hypersetup{
    colorlinks=true,
    urlcolor = blue, 
    linkcolor = blue, 
    citecolor = red 
}
\usepackage{url}            
\usepackage{booktabs}       
\usepackage{amsfonts}       
\usepackage{amsmath,stackengine}
\usepackage{nicefrac}       
\usepackage{microtype}      
\usepackage[table,dvipsnames]{xcolor}         
\usepackage{mathtools}
\usepackage{svg}            
\usepackage{caption}
\usepackage{subcaption}
\usepackage{mathpartir}     

\usepackage[conf={restate,no link to proof}]{proof-at-the-end}
\usepackage{comment}
\usepackage[normalem]{ulem}
\usepackage[linesnumbered,ruled,vlined]{algorithm2e}
\usepackage{enumitem}

\usepackage{color} 
\usepackage{xargs} 
\usepackage{xspace} 
\usepackage{minted}

\usepackage{accents}
\usepackage{booktabs}
\usepackage{makecell}
\usepackage{multirow}
\usepackage[normalem]{ulem}
\usepackage{cancel}

\usepackage{tabularray}
\DefTblrTemplate{contfoot-text}{normal}{}
\SetTblrTemplate{contfoot-text}{normal}

\DefTblrTemplate{conthead-text}{normal}{}
\SetTblrTemplate{conthead-text}{normal}

\DefTblrTemplate{capcont}{default}{}

\usepackage{pifont}
\usepackage{tikz} 
\usetikzlibrary{shapes,positioning,calc}
\tikzset{every picture/.style={remember picture}}
\usepackage{colortbl}

\colorlet{lightgray}{gray!20}
\newlength\stextwidth

\usepackage{listings} 
\newtheorem{theorem}{Theorem}[section]

\newtheorem{definition}[theorem]{Definition}

\newtheorem{lemma}[theorem]{Lemma}

\newtheorem{example}[theorem]{Example}

\allowdisplaybreaks

\def\notationcolor{black} 

\newcommand{\notation}[2]{\newcommand{#1}{{\textcolor{\notationcolor}{\ensuremath{#2}}}\xspace}}

\newcommand{\term}[2]{\newcommand{#1}{\textcolor{\notationcolor}{#2}\xspace}}

\DeclarePairedDelimiter\set{\{}{\}}            

\term{\priv}{Differential Privacy}
\term{\ekey}{entity key}

\term{\toolname}{\textsc{DP4SQL}}
\term{\pinq}{\textsc{PINQ}}
\term{\flex}{\textsc{FLEX}}
\term{\googledp}{\textsc{GoogleDP}}
\term{\privatesql}{\textsc{PrivateSQL}}
\term{\dopsql}{\textsc{DOP-SQL}}
\term{\opendptumult}{\textsc{OpenDP-Tumult}}

\notation{\mech}{\mathcal{M}}  
\notation{\att}{A} 
\notation{\atts}{\mathcal{A}} 
\notation{\rel}{R} 
\notation{\rels}{\mathcal{R}} 
\notation{\ent}{R^\star} 
\notation{\enttbl}{\tbl^\star} 
\notation{\entrec}{e} 
\notation{\ents}{\mathcal{E}} 
\notation{\schema}{\mathbb{S}} 
\notation{\db}{D} 
\notation{\dbs}{\mathcal{D}} 
\notation{\tbl}{T}
\notation{\record}{r}
\notation{\raexpr}{S}   
\notation{\query}{Q}
\notation{\view}{V} 
\notation{\plabel}{\mathbb{P}}
\notation{\slabel}{\mathbb{S}}
\notation{\constr}{\mathcal{C}} 
\notation{\keyconstr}{\mathcal{K}} 
\notation{\refconstr}{\mathcal{I}} 
\notation{\nulconstr}{\mathcal{N}} 
\notation{\key}{K}
\notation{\trans}{\mathcal{T}}
\notation{\parents}{\ensuremath{P_\schema}}
\notation{\world}{\mathcal{W}}

\notation{\pks}{\mathcal{PK}} 
\notation{\pk}{PK} 
\notation{\fks}{\mathcal{FK}} 
\notation{\fk}{FK} 
\notation{\owner}{\mathcal{O}}
\notation{\refers}{\ensuremath{\rightarrow\cdots\rightarrow}} 

\notation{\public}{Public}
\notation{\powners}{powners} 

\notation{\select}{\ensuremath{\sigma}}
\notation{\pred}{\ensuremath{\varphi}}
\notation{\proj}{\ensuremath{\pi}}
\notation{\join}{\ensuremath{\Join}}
\notation{\ljoin}{\ensuremath{\ltimes}}
\notation{\rjoin}{\ensuremath{\rtimes}}
\notation{\tcjoin}{\ensuremath{\Join_{\texttt{tc}}}}
\notation{\rename}{\ensuremath{\rho}}
\notation{\project}{\ensuremath{\pi}}
\notation{\predicate}{\ensuremath{\varphi}}
\notation{\distinct}{\ensuremath{\kappa}}
\notation{\groupby}{\ensuremath{\gamma}}
\notation{\ojoin}{\setbox0=\hbox{$\Join$}
  \rule[.02ex]{.25em}{.4pt}\llap{\rule[1.05ex]{.25em}{.4pt}}}
\notation{\leftouterjoin}{\mathbin{\ojoin\mkern-7.8mu\Join}} 
\notation{\trunc}{\tau}

\notation{\attr}{\ensuremath{\operatorname{attr}}}  
\notation{\privatt}{\texttt{private}}  
\notation{\diff}{\ensuremath{\operatorname{diff}}}  
\notation{\dom}{\ensuremath{\operatorname{dom}}} 
\notation{\cnt}{\ensuremath{\texttt{CNT}}}  
\notation{\avg}{\ensuremath{\texttt{AVG}}}    
\notation{\fsum}{\ensuremath{\texttt{SUM}}}    
\notation{\med}{\ensuremath{\operatorname{MED}}}    
\notation{\range}{\ensuremath{\texttt{Range}}}  
\notation{\fmax}{\ensuremath{\texttt{MAX}}}  

\newcommand{\adel}[2]{\ensuremath{\texttt{Del}_{#1}}}
\newcommand{\aadd}[2]{\ensuremath{\texttt{Add}_{#1}}}

\newcommand{\arep}[2]{\ensuremath{\texttt{Rep}_{#1}^{#2}}}

\notation{\lbl}{\ensuremath{\operatorname{\mathcal{L}}}}  
\notation{\ppolicy}{\ensuremath{\operatorname{\mathcal{P}}}}  
\notation{\tyenv}{\ensuremath{\operatorname{\Gamma}}}  
\notation{\atyenv}{\ensuremath{\operatorname{\Delta}}}  

\notation{\mf}{\texttt{mf}} 
\notation{\mmf}{\texttt{mmf}} 
\notation{\own}{\ensuremath{\otimes}} 
\notation{\mown}{\ensuremath{\hat\otimes}} 
\notation{\po}{\ensuremath{\preceq}}  
\notation{\lub}{\ensuremath{\wedge}}  
\notation{\glb}{\ensuremath{\vee}}  
\newcommand{\pdel}{\ensuremath{\mathbb{DEL}}} 
\newcommand{\prep}{\ensuremath{\mathbb{REP}}} 
\newcommand{\ppub}{\ensuremath{\mathbb{PUB}}} 

\notation{\nbrs}{\mathcal{N}}  
\notation{\nbr}{N}  
\notation{\action}{\alpha}  

\notation{\ownarrow}{\stackon[-1pt]{\ensuremath{\rightarrow}}{\text{\tiny{ow }}}}
\notation{\fkarrow}{\stackon[-1pt]{\ensuremath{\rightarrow}}{\text{\tiny{fk }}}}
\notation{\trownarrow}{\stackon[-1pt]{\ensuremath{\twoheadrightarrow}}{\text{\tiny{ow\phantom{...}}}}}
\notation{\trfkarrow}{\stackon[-1pt]{\ensuremath{\twoheadrightarrow}}{\text{\tiny{fk\phantom{...}}}}}

\AtBeginDocument{%
  }

\setcopyright{acmlicensed}
\copyrightyear{2018}
\acmYear{2018}
\acmDOI{XXXXXXX.XXXXXXX}
\acmConference[]{}{}{}

\begin{document}

\title{\toolname: Differentially Private SQL with Flexible Privacy Policies}

\author{Andrew Cascio}
\affiliation{%
  \institution{Duke University}
  \state{North Carolina}
  \country{USA}}
\email{ac940@duke.edu}

\author{KinChin Tong}
\affiliation{%
  \institution{Binghamton University}
  \state{New York}
  \country{USA}}
\email{ktong1@binghamton.edu}

\author{Daniel Kifer}
\affiliation{%
  \institution{Penn State University}
  \state{Pennsylvania}
  \country{USA}}
\email{dkifer@cse.psu.edu}

\author{Zeyu Ding}
\affiliation{%
  \institution{Binghamton University}
  \state{New York}
  \country{USA}}
\email{dding1@binghamton.edu}

\author{Danfeng Zhang}
\affiliation{%
  \institution{Duke University}
  \state{North Carolina}
  \country{USA}}
\email{dz132@duke.edu}


\begin{abstract}
The plausible deniability model of differential privacy for single-table datasets
is well-understood. However, applying differential privacy to relational
databases is much trickier: each application needs flexibility in specifying 
the pieces of information about an entity, spread across multiple relations, that require plausible
deniability guarantees. Existing differentially private SQL systems only support
rigid privacy policies. Even seemingly small changes, such as specifying that some
tables need to protect the existence of records while others only need to protect
the record contents, require significant manual effort in updating their privacy accountants
and proving their correctness.

One example of a challenge is the presence of partially public data. Public columns in a table
(e.g., faculty names in a university dataset and partial course enrollment information) 
can cause some queries to require more noise (compared to fully private data), while others require less noise. 
This kind of reasoning is not supported in existing
systems. Another example is when different parts of records (e.g., demographics, financial data) require different
levels of privacy protection. Again, existing differentially private SQL systems need to rewrite their
rules for calculating query stability in order to support such a feature.
This paper presents \toolname, a differentially private SQL system that allows data curators to better
customize the plausible deniability requirements for their relational databases. This avoids the drawbacks of the
``one-size-fits-all'' systems that would either underprotect the data or inject too much noise into query answers.
\end{abstract}




\maketitle

\section{Introduction}
Differential privacy (DP)~\cite{diffpbook, dwork06Calibrating} is a gold standard for creating mechanisms (algorithms) that generate publicly-releasable data products from confidential datasets, while protecting the private information in those datasets. It has an ever-increasing list of real-world deployments, including the U.S. Census Bureau~\cite{tdahdsr,ashwin08:map}, Uber~\cite{FLEX,chorus}, Apple~\cite{appledpscale}, Facebook~\cite{fburlshares}, Microsoft~\cite{DingKY17}, and Google~\cite{rappor,tensorflowprivacy,wilson2019differentially}. 

Using differential privacy is surprisingly complex. It requires specialized knowledge to design mechanisms that produce useful data products while satisfying the mathematical requirements of DP that guarantee privacy protection.
As a result, there is strong interest in creating DP  platforms \cite{opendp,gupt,Roy10Airavat}, especially differentially private SQL systems that ingest a non-expert's SQL queries and produce accurate, privacy-preserving answers
~\cite{PINQ,FLEX,wilson2019differentially,kotsogiannis2019privatesql,berghel2022tumult,dop-sql,chorus,wpinq,atlantic,pioneer}.

However, using such systems does \emph{not} necessarily mean that DP is used properly.
Applying differential privacy correctly requires specifying (1) an appropriate application-dependent plausible deniability model that specifies what pieces of information need to be indistinguishable from each other and (2) the strength of the indistinguishability. 
The plausible deniability model is an open question for DP SQL systems, and is the topic of this paper.
Meanwhile, the strength of the indistinguishability guarantee is a well-understood mathematical concept that links privacy parameters, like the famous $\epsilon$, to limitations on an attacker's ability \cite{dwork06Calibrating,zcdp,renyidp,wasserman2010statistical,fdp} to use the output of a mechanism to make guesses about the information that should be plausibly deniable.  

 Each existing DP SQL system \cite{PINQ,FLEX,wilson2019differentially,kotsogiannis2019privatesql,berghel2022tumult,dop-sql,chorus,wpinq,atlantic,pioneer} provides their own hard-coded, often competing,
plausible deniability models. A mismatch between the model and the application requirements could leave the data under-protected or over-protected, with incorrect amount of noise added to query answers. However, in existing systems, data administrators
cannot customize the plausible deniability model. There are several reasons for this limitation. (1) The plausible deniability model is tightly integrated into the privacy calculus used by those systems to guarantee DP. Even small changes to the model would require rewriting and manually proving the correctness of the privacy accounting rules. (2) Real-world requirements are so complex that requiring a data administrator to specify a plausible deniability model is a daunting task (even for data administrators with deep expertise in privacy technology). (3) The way plausible deniability is specified in differential privacy---the ``neighbor relation''---is extremely low-level and cumbersome.

A neighbor relation $\nbrs$ is a set of pairs of databases. If a database pair $(\db, \db')\in\nbrs$, it means that an attacker should have difficulty in determining whether the public data products were created from $\db$ or $\db'$. The difference in contents between $\db$ and $\db'$ is a piece of information that gets plausible deniability guarantees. For example, if the database schema contains only one table and every person can contribute only one record, then the appropriate relation $\nbrs_1$ is \emph{unbounded neighbors}: $(\db,\db')\in\nbrs_1$ if and only if $\db'$ can be obtained from $\db$ by the removal or addition of an arbitrary record. Hence the existence of a record gets plausible deniability. However, if the size of this table is publicly known, then the appropriate relation $\nbrs_2$ is \emph{bounded neighbors}: $(\db,\db')\in\nbrs_2$ if and only if $\db'$ can be obtained from $\db$ by \emph{replacing} one record. Under DP, the same query gets different noise when using the neighbor relations $\nbrs_1$ vs. $\nbrs_2$. 

Incidentally, it is worth noting that no existing DP SQL system can support situations where some tables in a database have publicly known sizes while others do not. Even such a seemingly small detail would require re-defining systems semantics, re-writing their privacy accountants and re-proving their correctness.

Common situations in practice introduce even more complexity.
Consider our running example of a simplified university database schema 
(Figure \ref{fig:university}) consisting of 5 relations. Student(\textbf{uid}, \emph{name}, \emph{major}) lists students enrolled in the university (table $\tbl_1$), and Faculty(\textbf{fid}, \emph{name}, \emph{salary}, \emph{age}) lists faculty members (table $\tbl_2$). These are the 2 types of entities that would need privacy protection. Scholarship(\textbf{aid}, uid, \emph{amount}) provides scholarship information (table $\tbl_3$), Section(\textbf{sid}, fid, \emph{title}) lists the course sections (table $\tbl_4$), and Enrollment(\textbf{eid}, sid, uid, \emph{grade}, \emph{review}) has the enrollment, grade, and student review information for each section (table $\tbl_5$). Complicating matters is the mixed sensitivity of information, even in the same table. For example, faculty name is public, faculty demographics are somewhat sensitive, and financial information is extremely sensitive.
The course section information is public. The students table is fully private. The size of the enrollment in each section (i.e., the section id (sid) column in $\tbl_5$) is public, the grades in the same table are private, and the university may wish to make review scores public while protecting the association between the review and students who gave the review.

To what degree can existing systems support such a scenario? Systems that explicitly track record ownership throughout query execution (e.g., \cite{wilson2019differentially}) cannot handle databases where different entities interact (e.g., faculty assigning grades to students). Some work cannot handle databases with foreign keys (i.e., it is impossible for them to provide plausible deniability for a student and all records owned by the student across different tables) \cite{PINQ,FLEX}. The most sophisticated privacy model in a DP SQL system \cite{kotsogiannis2019privatesql} handles those two cases but cannot support
fine-grained reasoning about different columns (e.g., a query about the number of highly-paid faculty who teach AI courses should require relatively more noise than a query about the number of male faculty who teach AI courses), cannot reason about tables whose sizes are known, and cannot reason about public columns or protect associations (e.g., who left the review, when the review score is public).

Our approach towards more expressive privacy policies, and contributions of the paper are the following:
\begin{itemize}[leftmargin=*]
    \item We propose a simple, high-level column labeling framework that allows a data administrator to specify multiple plausible deniability requirements for each type of entity in a database with foreign keys. Although conceptually simple, it provides support for situations where table sizes are known, some columns are public, other columns need more protection, etc.
    
    \item We propose a flexible, lower-level, plausible-deniability-action framework for specifying plausible deniability requirements in a way that is more suitable for automated reasoning. We develop an inference system with production rules that map the data administrator's column labeling into this more complex lower-level policy specification. This lower-level specification can be mapped into a neighbor relation, and this allows our inference system to further reason soundly about the stability of relational algebra queries, involving joins and aggregations, in order to compute how much noise must be added to the query answers.
    \item We implement this multi-level system, where a high-level human-friendly specification is translated into an expressive automated-reasoning-friendly specification, as a tool which we call \toolname. It is implemented as database middleware that intercepts SQL queries and determines how much noise needs to be added to satisfy all of the plausible deniability policies.
    \item Experiments on TPC-H ~\cite{tpc-h} and a case study on flexible privacy policies validate our approach and show that competing work does not add appropriate noise levels to SQL queries.
\end{itemize}

\begin{figure}
\centering
\includegraphics[width=\linewidth]{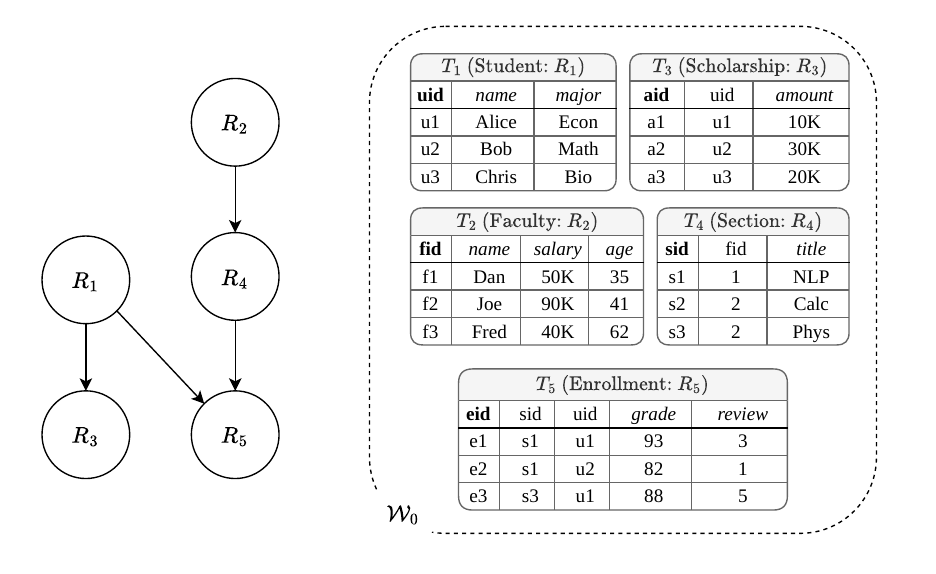}
\caption{The data ownership graph of a university schema (left) and an instance of the same schema, where $\tbl_i$ is the table for relation $\rel_i$ in a hypothetical world $\world_0$ (right).}
\label{fig:university}
\end{figure}

\section{Related Work}
\label{sec:related-work-multitable}
A central challenge for differentially private SQL is supporting practical
\emph{multi-table} queries—especially join-heavy workloads—without either (1)
unsound privacy accounting or (2) overly conservative sensitivity bounds that
destroy utility. Joins can amplify an individual’s contribution across multiple
relations, and this amplification depends on schema constraints (e.g., foreign
keys), data distributions (e.g., fanout), and query structure (e.g., chains of
joins, self-joins, and joins composed with grouping).

\pinq~\cite{PINQ} introduced the idea of integrating privacy accounting into a query-like
programming interface, enforcing bounded contribution and tracking stability of
transformations to calibrate noise.
Following systems in this line focus on expressing analyses as pipelines of
relational operators and ensuring that the accumulated stability remains bounded.
However, general join patterns remain challenging because contribution can grow
rapidly with fanout unless additional structural assumptions (e.g., key-joins) or
explicit clipping/truncation are imposed.

\flex~\cite{FLEX} advanced the state-of-the-art by
providing local sensitivity upper bounds for SQL with reasoning that leverages maximum-frequency (fanout)
bounds to control join amplification.
These approaches are effective for many join patterns but are typically tied to
a fixed privacy interpretation (commonly, under the bounded neighbors setting)
and a table-level policy model. As a result, they often treat
entire records as uniformly private, which can force noise even when only a
subset of attributes (or only specific associations across tables) require
protection.

\privatesql~\cite{kotsogiannis2019privatesql} formalized an entity-level neighboring relation for
\emph{multi-relational} databases with referential constraints and developed a
sensitivity analysis that is cognizant of join structure and ownership induced
by foreign keys.
This was an important step beyond single-table or record-level models because
it aligns privacy with real-world entities that own records in multiple tables.
However, \privatesql’s policy model is intentionally simple: it primarily targets
unbounded (delete-one-entity) semantics and cannot directly express
attribute-level policies where some columns are public and others are private,
nor can it simultaneously mix bounded and unbounded notions of plausibility
across relations.

\section{Notation and Background}\label{sec:notation}

A relational database schema is a pair $\schema=(\rels, \constr)$ where $\rels=\set{\rel_1, \ldots, \rel_n}$ is a set of relations providing metadata about tables in a database, and $\constr$ is a set of integrity constraints that specify foreign key restrictions on the relations. Each relation $\rel$ contains a finite set of attributes $\attr(\rel)=\set{\pk, \att_1, \ldots, \att_m}$\footnote{We require the primary key, $\pk$, to be semantically independent of the data. This assumption is not restrictive, as one can always introduce a unique random identifier as the primary key.}.
The set of possible values for an attribute $\att$ is called the \emph{domain} of $\att$ and denoted by $\dom(\att)$. We also write $\dom(\rel)=\dom(\pk)\times\dom(\att_1)\times\cdots\times\dom(\att_m)$. To distinguish attributes from different relations we write $\rel.\att$ to indicate that $\att$ is an attribute of $\rel$.

\paragraph{Relations, Tables, and Records} 
An instance of a relation $\rel$ is a \emph{table}, denoted $\tbl\subset \dom(\rel)$. Each element $\record \in \tbl$ is called a \emph{row} or a \emph{record}. We use $\record.\att$ to denote the component of $\record$ that corresponds to attribute $\att$, and $\tbl.\att$ for multiset $\set{\record.\att \mid \record\in\tbl}$. For a schema $\mathbb{S}=(\mathcal{R},\mathcal{C})$, $\text{dom}(\mathbb{S})$ is the set of database instances of $\mathcal{R}$ (i.e., tables of relations in $\mathcal{R}$) that satisfy $\mathcal{C}$. An instance of $\schema$ is a \emph{database} $\db\in\dom(\schema)$. 

\paragraph{Integrity Constraints}
The set of integrity constraints $\constr$ specify links between relations. If $\rel.\fk_{\rel'}$ is the foreign key attribute for $\rel'.\pk$, the primary key of $\rel'$, then we write $\rel.\fk_{\rel'}\fkarrow\rel'.\pk$ and say $\rel$ references $\rel'$. For simplicity, we may also write $\rel\fkarrow\rel'$. Moreover, if $\record\in\tbl,\record'\in\tbl'$ and $\record.\fk_{\rel'}=\record'.\pk$, then we write $\record\fkarrow\record'$.  
If there is a path from $\rel$ to $\rel'$ that follows foreign keys, we say that $\rel$ \emph{transitively refers} to $\rel'$. Formally:

\begin{definition}[Transitive Referral]
    A relation $\rel$ transitively refers to a relation $\rel'$ if $\rel\fkarrow\rel'$ or (recursively) if there exists a relation $\rel''$ such that $\rel\fkarrow\rel''$ and $\rel''$ transitively refers to $\rel'$. Similarly, a record $\record\in\tbl$ transitively refers to a record $\record'\in\tbl'$ if $\record\fkarrow\record'$ or if there exists a record $\record''\in\tbl''$ such that $\record\fkarrow\record''$ and $\record''$ transitively refers to $\record'$. In both cases, we use  $\trfkarrow$ to denote transitive referral relation.
\end{definition}
As standard, we require foreign keys to be acyclic (i.e., no relation transitively refers to itself). 

\subsection{Data Ownership Graph}
The \emph{data ownership graph} (e.g., Figure \ref{fig:university}) is a key tool for understanding which entities may potentially have ownership of which records. We say a record $\record'\in\tbl'$  \emph{owns} $\record\in\tbl$, denoted by $\record'\ownarrow\record$, if $\record$ has a foreign key to $\record'$ (i.e., the direction of ownership is the \textbf{reverse} of the direction of foreign keys). Similarly, we say $\record'$ transitively owns $\record$, denoted by $\record'\trownarrow \record$ if and only if $\record\trfkarrow \record'$. We extend this notation to relations in the obvious way. The ownership arrows between relations form the data ownership graph:

\begin{definition}[Data Ownership Graph] \label{def:dependencygraph}
    The \emph{data ownership graph} of a relational schema $\schema=(\rels,\constr)$ is a directed graph $G(\schema)=(V,E)$ where $V=\rels$ and $(\rel', \rel)\in E$ if and only if $\exists~\rel.{\att} \fkarrow \rel'.\pk \in \constr$. 
\end{definition}

We also define the records in $\tbl_2$ owned by a record $\record_1\in\tbl_1$:
\begin{definition}[Ownership]\label{def:ownership}
    Let $\schema=(\rels,\constr)$, let $R_1,R_2\in\mathcal{R}$ be any two relations, and let $\tbl_1\subset\dom(\rel_1)$ and $\tbl_2\subset\dom(\rel_2)$. The records that $\record_1\in\tbl_1$ owns in $\tbl_2$ is defined as:
    \begin{align*}
        \own(\tbl_1,\tbl_2,\record_1)=\{\record_2\in\tbl_2\mid\record_1\trownarrow\record_2\}
    \end{align*}
\end{definition}

Some relations (e.g., Faculty and Student) are called \emph{entity relations} because they define entities who need privacy protection. Let $\ents$ be the set of entity relations. Relations not in $\ents$ are called \emph{non-entity relations}. Entity relations are ``roots'' in the data ownership graph
and they can transitively own records in non-entity relations. 
For example, in the data ownership graph in Figure \ref{fig:university}, $\rel_1$ (with table $\tbl_1$) is the entity relation Student and it has an ownership arrow to Enrollment ($\rel_5$) because of the foreign key going the other way. The entity relation Faculty ($\rel_2$) has an ownership arrow to Section ($\rel_4$), which has an ownership arrow to Enrollment. So Faculty also transitively owns Enrollment.
For convenience, our notation is summarized in Table~\ref{tab:notation} in order of appearance.

\begin{table}
\centering
\caption{Table of Notation.}\label{tab:notation}
\begin{tabular}{|l@{\hspace{-2em}}r|}
\hline
\multicolumn{2}{|l|}{\textbf{Notation from Section~\ref{sec:notation}}} \\
\color{black}
$\schema$:& Pair $(\rels, \constr)$ of set of relations and integrity constraints.\\
$\rels$: & A set $\set{\rel_1, \ldots, \rel_n}$ of relations.\\
$\constr$: & Database integrity constraints (i.e., foreign key links).\\
$\rel$: & A relation with attributes $\attr(\rel)=\set{\pk,\att_1, \ldots, \att_m}$.\\
$\tbl$: & A table; concrete instance of a relation.\\
$\db:$ & A database; concrete instance of a schema.\\
$\fkarrow, \trfkarrow$: & Indicates foreign key and transitive foreign key.\\
$\ownarrow, \trownarrow$: & Indicates ownership and transitive ownership.\\
$\own(\tbl_1,\tbl_2,\record_1)$: & Records in $\tbl_2$ transitively owned by  $\record_1\in\tbl_1$. \\
\multicolumn{2}{|l|}{\textbf{Notation from Section~\ref{sec:privacy_model}}} \\
$\ent:$ & A distinguished entity relation.\\
$\ppolicy_{\schema,\ent}$: & A privacy policy for schema $\schema$ and entity $\ent$.\\
$\nbrs(\ppolicy_{\schema,\ent})$: & Set of neighboring database pairs w.r.t. policy.\\
\multicolumn{2}{|l|}{\textbf{Notation from Section~\ref{sec:sensitivity}}} \\
$\mf(\tbl.\att)$: & Frequency of most frequent value of $\tbl.\att$. \\
$\mmf(\rel.\att)$: & Computed upper bound on $\mf(\tbl.\att)$. \\
$\mown(\rel_1,\rel_2)$: & Maximum number of records $\rel_1$ can own in $\rel_2$.\\
$\Delta_{\ppolicy_{\schema,\ent}}(Q)$: & Global sensitivity of query $Q$ w.r.t. policy. \\
$\hat\Delta_{\ppolicy_{\schema,\ent}}(\query)$: & Inferred upper bound on global sensitivity. \\
\hline
\end{tabular}
\end{table}

\subsection{Differential Privacy, Counterfactual Worlds, and Plausible Deniability}
The plausible deniability model of differential privacy is often viewed through the lens
of counterfactual reasoning \cite{tschantz2020sok}. Any possible database $\db$ that 
satisfies the integrity constraints and is consistent with public knowledge is part of
a \emph{hypothetical world} $\world$. For each entity $x_i$ (e.g., a student or a faculty),
there is at least one \emph{counterfactual} version of this world $\world^{(-x_i)}$ whose corresponding
database $\db^{(-x_i)}$ satisfies the integrity constraints and is consistent with the 
same public knowledge, but from which private information about $x_i$ has been scrubbed. Thus
$\db^{(-x_i)}$ is the privacy-preserving baseline for $x_i$ and differential privacy tries
to ensure that an attacker's inference about $x_i$ in the world $\world$ is nearly the
same as in the world $\world^{(-x_i)}$.

Each pair $(\db,\db^{(-x_i)})$ and $(\db^{(-x_i)},\db)$ of hypothetical $\db$ and counterfactual $\db^{(-x_i)}$ databases is called a
pair of neighbors. Properly defining neighboring databases is subtle and challenging, as we explain
in Section \ref{sec:motivation}. Though once a collection of neighboring pairs has been specified, differential
privacy can be defined:

\begin{definition}[Differential Privacy]
\label{def:dp_with_neighbors}
Given a collection $\nbrs$ of neighboring database pairs and a privacy parameter $\epsilon\geq 0$,
    a mechanism $\mech:\dom(\schema)\longrightarrow\Omega$ is $\epsilon$-differentially private if for every set of outputs $O\subseteq\Omega$ and every $(\db,\db')\in\nbrs$:
    \begin{displaymath}
        e^{-\epsilon}Pr[\mathcal{M}(\mathcal{D}')\in O]\leq
        Pr[\mech(\db)\in O]\leq e^\epsilon Pr[\mech(\db')\in O]
    \end{displaymath}
\end{definition}

\section{Motivation}
\label{sec:motivation}
To motivate the challenges in defining counterfactual worlds, and to see where existing DP SQL frameworks fall short, let us return to the university schema of Figure \ref{fig:university}. Suppose this fictitious university has decided on the following policy, stated informally as:
\begin{enumerate}[leftmargin=*]
    \item Faculty name from relation $\rel_2$ is public.
    \item All other faculty information (demographics, salary) should be private but salary should have stronger protections.
    \item The review scores from relation $\rel_5$ that students assign to faculty should be public, but the association (i.e., which student assigned which review score) is private.
    \item The number of scholarships from relation $\rel_3$ is public, but the award amount and the who the recipients are is private.
    \item The sections $\rel_4$ relation is public, but information about which students are in which section is private. Thus the section id (sid) column in $\rel_5$ is public.
    \item All other information about students is private.
\end{enumerate}

We next consider appropriate counterfactual worlds and then analyze the shortcomings of the privacy models used in prior work.
\subsubsection*{Counterfactuals for Faculty.}
 $\rel_2$ is the entity relation for faculty and faculty have ownership of records in $\rel_4$ (sections they teach) and $\rel_5$ (since they assign grades). Because of the public information (name column of $\rel_2$, all columns of $\rel_4$, and review column of $\rel_5$), one cannot form a privacy-preserving counterfactual world that deletes any records owned or transitively owned by a faculty member. The best that can be done is to replace the contents of demographics and salary in $\rel_2$ with different values, and replace the grades in $\rel_5$. This gives multiple counterfactual worlds for each faculty member, where each world corresponds to a different setting of the modifiable attributes in records owned by a faculty.

\subsubsection*{Extra Protections for Salary.}
To provide extra protections for a faculty salary, one can  create additional counterfactual worlds for each faculty  by (counter-intuitively) altering just one faculty salary. The extra protection comes from requiring a smaller $\epsilon$ to be used for these neighbors. Thus, this is a relatively simple situation that only requires fine-grained column-based reasoning from a DP SQL system (however, this is not supported by prior work).

\subsubsection*{Counterfactuals for Students.}
The entity relation for students is $\rel_1$ and it owns the enrollments relation $\rel_5$. A counterfactual world for a specific student would drop that student from the students table $\tbl_1$ but would only be able to reassign their grade and student id (uid) for all sections the student is enrolled in (i.e., enrollment records cannot be dropped, but parts of them can be modified).

\subsection{Limitations in Prior Models}
PINQ \cite{PINQ} and Flex \cite{FLEX} were the earliest forays into DP SQL systems and so had relatively simple privacy models. 
These models inherently cannot reason about foreign keys and so cannot support these use-cases. They only support counterfactual worlds in which exactly 1 table differs by 1 record. Thus their counterfactual worlds cannot serve as privacy-preserving baselines and generally they will under-protect query answers---they add just enough noise to mask the existence of one record, rather than enough noise to mask all information owned by an entity.

The pioneering work of Wilson et al. \cite{wilson2019differentially} tracks the ownership of records through the query processing pipeline but it cannot support databases where entities interact and records can be owned by multiple entities (like the Enrollment relation $\rel_5$ in Figure \ref{fig:university}), and cannot support fine-grained reasoning about columns.

The most sophisticated DP SQL model was proposed by Kotsogiannis et al. \cite{kotsogiannis2019privatesql} and it can handle both foreign keys and multiple types of interacting entities. However, it cannot perform fine-grained reasoning about columns or about tables whose sizes cannot change in counterfactual worlds (like the Faculty table). It also cannot reason about associations (e.g., hiding associations between students and their reviews). So for some queries, it can unnecessarily add too much noise, and for others it can add too little noise. 

Consider the query \texttt{SELECT s.amount, AVG(e.grade) FROM Scholarship s, Enrollment e WHERE s.uid = e.uid GROUP BY s.amount}. This query calculates the average grade of students for different scholarship amounts. To protect a student, \privatesql~\cite{kotsogiannis2019privatesql} can only consider counterfactual worlds where all records with the same uid (across Student, Scholarship, and Enrollment) are dropped. This would cause 1 student to drop out of the join and \privatesql adds just enough noise to cover such a change. However, such counterfactual worlds are not consistent with public knowledge. Instead, a consistent counterfactual world requires the foreign keys from Scholarship and Enrollment that point to the dropped student to be changed to someone else. Thus, a student's result would drop out of the join (affecting one group in the group-by), but the changed foreign keys in the Scholarship and Enrollment tables could each link to a new person, therefore affecting many other groups. Thus, to make a world indistinguishable from its counterfactual world, much more noise must be added. Hence, \textbf{it is an issue of correctness}---if the privacy model cannot be customized, some queries will receive too much noise while others will receive too little noise.

\section{High-level Privacy Specifications}\label{sec:privacy_model}
In our proposed framework, \toolname, a data administrator specifies a set of privacy policies for for each \emph{entity relation},
such as students and faculty, each associated with a level of privacy strength (i.e., a  $\epsilon$ parameter for pure differential privacy). A query would be evaluated against
each entity relation and conservatively add the largest noise variance among all to ensure sufficient privacy protection for all entities. Hence, for the rest of this paper, we assume a distinguished entity relation $\ent\in\rels$ and show how to specify and enforce a privacy policy for $\ent$.

We start with a high-level, ``privacy labels'' framework  that a data administrator can use to specify a policy. Then in Section \ref{sec:sensitivity}, we propose a lower-level ``plausible deniability action'' framework that is more suitable for automated reasoning---we show how to translate the labeling policy into a set of plausible deniability actions and how to track query stability and sensitivity so that the right amount of noise can be added.

\subsection{High-Level Privacy Labels}

Each policy is associated with a distinguished entity relation $\ent\in\rels$ from the schema $\schema$ and assigns a label to each relation. Hence we denote a policy as $\ppolicy_{\schema,\ent}=\{P_1,\cdots,P_n\}$, where $P_i$ is the label for $\rel_i$. For convenience, we simply write $\ppolicy_{\schema,\ent}(\rel_i)=P_i$, and may omit the $\schema$ in the notation (i.e., $\ppolicy_{\ent}$) when clear from context. There are 3 choices of labels for a relation $\rel$ and its associated table $\tbl$:

\begin{itemize}[leftmargin=*]
    \item $\pdel$: Everything in the associated table $\tbl\subset\dom(\rel)$ (including its size) is private, and the existence of records is protected. The counterfactual worlds are created by deleting records. 
    \item $\prep^{\atts}$: The values of all attributes in $\atts$ are deemed private, but the \emph{existence} of a record in $\tbl$ is not protected. Hence, the table size is public (i.e., not protected by this policy) and attributes not in $\atts$ are also public. 
    Counterfactual worlds are created by modifying the values of the attributes in $\atts$. 
    \item $\ppub$: The entire relation is  public.  
\end{itemize}

Intuitively, the difference between $\pdel$ and $\prep^{\atts}$ policies lies in what “protecting an individual” means. $\pdel$ protects the existence of an entity in the dataset, while $\prep^{\atts}$ only hides attributes belonging to entities, not their presence. Hence, in a single-relation setting, the $\pdel$ policy is the same as using DP with unbounded neighbors, and $\prep^{\atts}$ is akin to bounded neighbors.

For example, if the number of students is private, we can give the relation the $\pdel$ label to protect the existence of students. If the number of students is public but their names are private, we can give the relation the $\prep^{\{\text{name}\}}$ label. If all student information is public (i.e., names, majors), we can give the relation the $\ppub$ label. Note that whenever an attribute $R.A$ is public, the size of $R$ is public as a consequence.

\subsection{Neighboring Databases}
\label{sec:neighbor}

The definition of neighboring databases can be constructed from the privacy labels. We use the intuition about counterfactual world properties, explained above in the context of individual relation labels, to define counterfactual
worlds that are consistent with ownership, public information and integrity constraints.

Given the schema $\schema=(\rels,\constr)$, database $\db=\{\tbl_1,\cdots,\tbl_n\}\in\dom(\schema)$, distinguished entity relation $\ent$,  its associated table $\tbl_\ent\in \db$, and privacy policy $\ppolicy_{\schema,\ent}=\{P_1,\cdots,P_n\}$, we next formally define the meaning of the privacy policy, in terms of what it means for $\db'=\{\tbl'_1,\cdots,\tbl'_n\}$ to be a neighbor of $\db$ under the policy. Intuitively the following definition requires that: (1) $\tbl_i$ and $\tbl'_i$ only differ in the records owned by a single entity $e\in\tbl_\ent$, (2) the difference in those records obeys policy $P_i$, and (3) all integrity constraints $\constr$ are satisfied. Recalling that $\own(\tbl_\ent,\tbl_i,e)$ is the set of records in $\tbl_i$ transitively owned by  entity record $e\in\tbl_\ent$, we define neighbors and DP for the policy as follows:

\begin{definition}[Neighboring Databases]
    Let $\schema=(\rels,\constr)$ be a relational schema and $\ppolicy_{\schema,\ent}=\{P_1,\cdots,P_n\}$ a privacy policy. We denote by $\nbrs(\ppolicy_{\schema,\ent})$ the set of pairs of databases $\db=\{\tbl_1,\cdots,\tbl_n\}$ and $\db'=\{\tbl'_1,\cdots,\tbl'_n\}$ such that $\forall(\db,\db')\in\nbrs(\ppolicy_{\schema,\ent})$:
    \begin{itemize}[leftmargin=*]
        \item $D$ and $\db'$ satisfy $\constr$, and

        \item $\tbl_\ent$ and $\tbl'_\ent$ differ by exactly 1 record $\entrec \in \tbl_\ent$, and

        
        \item if $O_i=\own(\tbl_\ent,\tbl_i,e)$, $\overline{O}_i=\tbl_i\setminus O_i$ and $P_i$ is...
        \begin{itemize}
           \item $\pdel$, then $T_i\bigcap T_i' = \overline{O}_i$ and $\tbl'_i \setminus \tbl_i=\emptyset$.
            
            \item $\prep^{\atts}$, then $|\tbl_i|=|\tbl_i'|$, $T_i\bigcap T_i' = \overline{O}_i$, and $\forall\record\in O_i.~\exists \record'\in \tbl'_i$ such that $\record$ only differs from $\record'$ by the values of attributes in $\atts$.

            \item $\ppub$, then $T_i'=T_i$.
        \end{itemize}
    \end{itemize}
\label{def:neighbor}
\end{definition}

Note that by instantiating the neighboring relation $\nbrs$ in Definition~\ref{def:dp_with_neighbors} with Definition~\ref{def:neighbor}, we have the formal differential privacy definition with \toolname.

\begin{definition}[DP for Privacy Policy]\label{def:policy_dp}
   Let $\schema$ be a schema, $\ppolicy_{\schema,\ent}$ be a privacy policy, and $\epsilon$ a privacy budget. A mechanism $\mech:\dom(\schema)\longrightarrow\Omega$ is $(\ppolicy_{\schema,\ent},\epsilon)$-differentially private if for every set of outputs $O\subseteq\Omega$ and $\forall(\db,\db')\in\nbrs(\ppolicy_{\schema,\ent})$:
   \begin{align*}
      e^{-\epsilon}Pr[\mech(\db')\in O]\leq Pr[\mech(\db)\in O]\leq e^\epsilon Pr[\mech(\db')\in O]
   \end{align*}
\end{definition}

\subsection{Checking Privacy Labels}\label{sec:typecheckissue}

The flexibility of our approach also makes it possible for a data administrator to create nonsensical labels.
Thus, when we translate privacy labels into the lower-level plausible deniability actions in Section \ref{sec:sensitivity}, we also include consistency checks (Section \ref{sec:typingrulesbase}). To help motivate
the need for those checks, we give an example of a nonsensical policy here.

Consider a schema with an entity relation $\rel_1(\textbf{eid}, A)$ having primary key \textbf{eid} and attribute $A$, and a non-entity relation $\rel_2(\textbf{id}, \text{eid}, B)$ with a primary key $\textbf{id}$, foreign key eid into $\rel_1$ and an attribute $B$. Suppose the labels are $\pdel$ for $\rel_1$ and $\ppub$ for $\rel_2$. This signals the intention to protect the existence of entities (records) in $\rel_1$. However, the public foreign key in $\rel_2$ already reveals existence of entities and simply deleting records from a table instance $\tbl_1\subset\dom(\rel_1)$ could result in an inconsistent database. The inference system would reject such a policy. Instead, the data administrator would have two choices: (1) either choose the label $\prep^{\{A\}}$ for $\rel_1$  to protect the attribute $A$ but not record existence or (2) choose the label $\prep^{\{\text{eid}\}}$ for $\rel_2$ to protect the foreign key from $\rel_2$ to $\rel_1$, but leave the rest of the attributes in $\rel_2$ public.

\section{Low-Level Privacy Specifications: The Plausible Deniability Actions}\label{sec:sensitivity}

To enforce DP for a privacy policy $\ppolicy_{\schema,\ent}$, a fundamental challenge is to analyze the \emph{sensitivity} of a given SQL query: to what extent can its output vary when evaluated on two neighboring databases? 
For example, consider the following query to the database in Figure~\ref{fig:university}: \emph{What is the total number of enrollments?} The sensitivity is the greatest change to the answer of this query between any two neighboring databases. For simplicity, consider the unbounded neighbors case, where all neighbors result from removing a single student from the database. In the worst case, the sensitivity is the maximum number of sections a student can be enrolled in by university policy. 

The \emph{stability} of a relational operator (distinct from, but closely related to sensitivity) is the greatest number of changed records that may result from applying the relational operator to a pair of neighboring databases. The distinction is that stability is a property of functions that output relations (e.g., relational operators), whereas sensitivity is a property of functions that output numbers (e.g., count queries). 
Historically, starting with PINQ~\cite{PINQ}, stability is defined as the number of records in the output of a relational operator that are affected by one record in its input. However, due to its more expressive privacy model, \textsc{DP4SQL} also needs to track which set of attributes is changed. Hence,  we  propose and use  a more general type of stability analysis called \emph{plausible deniability actions}, which denote the greatest change that may result from applying a relational operator to a pair of neighboring databases. 

We first introduce \emph{plausible deniability actions}, which help to reason about the effect of relational operators on neighboring databases (Section~\ref{sec:pda}), and then develop an inference system, formalized as a set of derivation rules, to automatically derive plausible deniability actions for SQL queries (Section~\ref{sec:inference_base}). 
Finally, we transform the actions to an upper bound on the sensitivity (Section~\ref{sec:sens_inference}), which determines how much noise must be added to the query answer.

\subsection{Plausible Deniability Actions}
\label{sec:pda}

For multi-relational databases, PrivateSQL ~\cite{kotsogiannis2019privatesql} defined the \emph{global sensitivity of a view} as the maximum number of distinct rows that may differ between a query's outputs when executed on any pair of neighboring databases. However, due to the flexible model of \toolname, we need to track more information (e.g., which set of attributes might change) during the stability analysis.

\toolname proposes the following set of 3 \emph{plausible deniability actions}. Informally, they describe the operations to be performed on a table $\tbl$ or view in one database that can transform it into a table or view in a neighboring database, as specified by the data administrator's privacy labels.

\begin{enumerate}
    \item $\aadd{a}{\rel}$: Add $a$ rows.
    
    \item $\adel{d}{\rel}$: Delete $d$ rows.
    
    \item $\arep{k}{\atts}$: Replace values of attributes $\atts$ in $k$ rows.
\end{enumerate}

Note that we do not explicitly introduce a no-op action for public relations as this can be modeled with $\arep{0}{\emptyset}$. Moreover, $\aadd{a}{}$ and $\adel{d}{}$ are always written together as the product action $\aadd{a}{}\times\adel{d}{}$, as some queries may require both the addition and deletion of records when joins are involved.

\subsection{Maximum Frequency}
\label{sec:mmf}

To soundly approximate all possible changes between two neighboring tables or views generated by SQL queries, we follow prior work \cite{FLEX, kotsogiannis2019privatesql} by assuming a static upper bound on the maximum frequency of each attribute in the schema. Such upper bounds are even used in work on truncation operators (e.g., \cite{dong2022r2t}) that try to reduce global sensitivity. In the case of tables/attributes that are public in all privacy policies (e.g., globally public), their maximum frequencies can be taken from the  data. 

Let $\rel$ be a relation, let $\tbl\subset\dom(\rel)$ be an associated table, and let $\att\in\attr(\rel)$ be an attribute. The (instance-specific) \emph{maximum frequency} $\mf(\tbl.\att)$ is the frequency of the most frequent value of $\att$ in table $\tbl$. We also assume an (instance-independent) upper bound on the maximum frequency among \emph{all possible tables}, which is denoted $\mmf(\rel.\att)$. 

\subsection{Maximum Ownership}
To derive plausible deniability actions, we need to compute an upper bound on the number of records that are owned by an entity $\entrec$. This is called the \emph{maximum ownership} of $\entrec$. Generally, if $\tbl_1,\tbl_2\in\db$, the maximum number of records owned in $\tbl_2$ by any record $\record\in\tbl_1$ can be statically computed as follows:

\begin{definition}[Maximum Ownership]\label{def:max_ownership}
    Let $\schema=(\rels,\constr)$ where $\rel_1,\rel_2\in\rels$. The maximum number of records a record $\record$ in an instance of $\rel_1$ can own in an instance of $\rel_2$ is defined as follows:
    \begin{align*}
        \mown(\rel_1,\rel_2)=
        \begin{cases}
            1 & \text{if $\rel_1=\rel_2$} \\
            \sum\limits_{\rel_2.\fk_i\fkarrow\rel_j.\pk} \mmf(\rel_2.\fk_i)\cdot\mown(\rel_1,\rel_j) & \text{otherwise}
        \end{cases}
    \end{align*}
\end{definition}

We show that maximum ownership is indeed an upper bound on the size of ownership with the following lemma.

\begin{lemma}[Correctness of Maximum Ownership]\label{lemma:maximum_ownership}
    Let $\schema=(\rels,\constr)$ be a schema, $\db\in\dom(\schema)$ any database, and $\rel_1,\rel_2\in\rels$ any two relations with corresponding instances $\tbl_1,\tbl_2\in\db$. Then,
    \begin{align*}
        \forall\record_1\in\tbl_1.~|\own(\tbl_1,\tbl_2,\record_1)|\leq\mown(\rel_1,\rel_2).
    \end{align*}
\end{lemma}

\subsection{Action Inference for Base Relations}\label{sec:typingrulesbase}
\label{sec:inference_base}

\begin{figure}
\framebox{\textbf{Inference Rules for a Distinguished Entity Relation}}
\centering
\begin{mathpar}
\inferrule*[right=(\textsc{E-Del})]
    {\ppolicy_{\ent}(\ent)=\pdel}
    {\vdash\ent: \aadd{0}{\emptyset} \times \adel{1}{\rel}}

\inferrule*[right=(\textsc{E-Pub})]
    {\ppolicy_{\ent}(\ent)=\ppub
    }
    {\vdash\ent: \arep{0}{\emptyset}}

\inferrule*[right=(\textsc{E-Rep})]
    {\ppolicy_{\ent}(\ent)=\prep^{\atts} \\ \atts\subset\attr(\ent) \\ \ent.\pk\notin\atts}
    {\vdash\ent: \arep{1}{\atts}}
\end{mathpar}
\framebox{\textbf{Inference Rules for Non-Entity Relations}}
\begin{mathpar}
\inferrule*[right=(\textsc{F-Del})]
    {\ppolicy_{\ent}(\rel)=\pdel
    }
    {\vdash\rel: \aadd{0}{} \times \adel{\mown(\ent,\rel)}{}}

\inferrule*[right=(\textsc{F-Pub})]
    {\ppolicy_{\ent}(\rel)=\ppub \\
    \forall i.~\rel\fkarrow\rel_i.~\neg\privatt(\rel_i.\pk)
    }
    {\vdash\rel: \arep{0}{\emptyset}}

\inferrule*[right=(\textsc{F-Rep})]
    {\ppolicy_{\ent}(\rel)=\prep^{\atts} \\ \atts\subset\attr(\rel) \\ \rel.\pk\notin\atts \\\\
    \forall i.~\rel\fkarrow\rel_i.~\privatt(\rel_i.\pk)\implies\rel.\fk_{\rel_i}\in\atts
    }
    {\vdash\rel: \arep{\mown(\ent,\rel)}{\atts}}
\end{mathpar}
\caption{Inference rules for base relations. $\ent$ is a distinguished entity relation and $\rel$ is any non-entity relation. The \privatt(\att) predicate is true when \att is protected by deletion or replacement under $\ppolicy_{\ent}$.}
\label{fig:base_relation_rules}
\end{figure}

We develop an inference system to automatically compute the plausible deniability actions, from privacy labels, for each base relation $\rel\in\rels$. The inference rules are shown in Figure~\ref{fig:base_relation_rules}. The rules are relatively simple; there is one rule for each privacy label for each kind of relation. Each rule states that the plausible deniability action below the horizontal line can be derived whenever all assumptions above the line are valid.

\paragraph{Distinguished Entity Relations}
For a distinguished entity relation $\ent$, plausible deniability actions are derived directly from the privacy labels themselves. If $\ppolicy_\ent(\ent)=\pdel$ (\textsc{E-Del}), then the action for $\ent$ is $\aadd{0}{}\times\adel{1}{}$, since the existence of a single entity is protected by deletion. If $\ppolicy_\ent(\ent)=\ppub$ (\textsc{E-Pub}), then the action for $\ent$ is $\arep{0}{\emptyset}$ since all attributes are considered public. If $\ppolicy_\ent(\ent)=\prep^{\atts}$ (\textsc{E-Rep}), then the action for $\ent$ is $\arep{1}{\atts}$ since the values of attributes in $\atts$ of a single entity are protected by replacement.

\paragraph{Non-Entity Relations}
For a non-entity relation $\rel\neq\ent$, deriving the plausible deniability action is slightly more complicated due to dependencies; $\rel$ may have foreign keys that refer to other relations. 

Let $\entrec$ be an entity record that requires protection. If $\ppolicy_\ent(\rel)=\pdel$ (\textsc{F-Del}), then the existence of $\entrec$ must be protected by deletion. In the worst case, $\entrec$ owns $\mown(\ent,\rel)$ records in $\rel$. Therefore, the action for $\rel$ is $\aadd{0}{}\times\adel{\mown(\ent,\rel)}{}$.

The nontrivial cases are when $\ppolicy_{\ent}(\rel)=\ppub$ or $\ppolicy_{\ent}(\rel)=\prep^{\atts}$. In the former case (\textsc{F-Pub}), the $\ppub$ policy requires that the entire relation $\rel$ is public. Therefore, all foreign keys are public, and all matching primary keys are required to be public, as required by the assumption: $\forall i.~\rel\fkarrow\rel_i.~\neg\privatt(\rel_i.\pk)$. Since the entire relation is public, the action for $\rel$ is $\arep{0}{\emptyset}$. In the latter case (\textsc{F-Rep}), the $\prep^{\atts}$ policy requires that the attributes in $\atts$, which may include foreign keys, are private. Therefore, each private primary key must be matched with a private foreign key, as required by the assumption: $\forall i.~\rel\fkarrow\rel_i.~\privatt(\rel_i.\pk)\implies\rel.\fk_{\rel_i}\in\atts$. Since $\entrec$ owns at most $\mown(\ent,\rel)$ records in $\rel$, the action for $\rel$ is $\arep{\mown(\ent,\rel)}{\atts}$. 
We note that these checks essentially ensure that foreign keys are at least as private as their corresponding primary keys. While they are not needed when $\ppolicy_{\ent}(\rel)=\pdel$ since all attributes are private by assumption, they are required for $\ppub$ and $\prep^{\atts}$ to rule out nonsensical labels discussed in Section~\ref{sec:typecheckissue}.

\subsection{SQL Queries and Relational Algebra}

Next, we extend the inference system to support an  expressive subset of SQL. Like \privatesql~\cite{kotsogiannis2019privatesql}, \toolname supports useful SQL operations such as \texttt{SELECT}, \texttt{WHERE}, \texttt{JOIN} (equijoins), and \texttt{GROUP BY} clauses, as well as intermediate counting aggregations and subqueries in the \texttt{WHERE} clause. In addition to \texttt{COUNT(*)}, 
\toolname also supports the $\texttt{SUM($\att$)}$ aggregation. 

\paragraph{Relational Algebra}

\begin{figure}
    \centering
    \begin{align*}
    \query&::=\groupby_{\atts}^{\cnt(*)}(\raexpr)\mid\groupby_{\atts}^{\fsum(\att)}(\raexpr) \\
    \raexpr&::=\rel\mid\select_\pred(\raexpr)\mid\project_\atts(\raexpr)\mid\groupby_{\atts}(\raexpr)\mid\groupby_{\atts}^{\cnt(*)}(\raexpr)\mid\raexpr_1\underset{\att_1=\att_2}{\join}\raexpr_2
\end{align*}
    \caption{Syntax of relational algebra supported by \toolname. $Q$ is a query. $\rel$ is any base relation, $\att$ is an attribute, $\atts$ is a set of attributes, and $\pred$ is a logical predicate over attributes.}
    \label{fig:ra_syntax}
\end{figure}

As standard, we use relational algebra to model the semantics of the supported SQL queries. The syntax is shown in Figure~\ref{fig:ra_syntax}. For a schema $\schema=(\rels,\constr)$, a relational algebra \emph{expression} $\raexpr$ is either a base relation $\rel\in\rels$ or a relational transformation. Supported transformations are the select ($\select$), project ($\project$), equijoin ($\underset{\att_1=\att_2}{\join}$), grouping ($\groupby_{\atts}$), and grouping with count ($\groupby^{\cnt(*)}_{\atts}$) operators. A \emph{view} $\raexpr(\db)$ is an instance of an expression $\raexpr$ for database $\db$. Lastly, a top-level query $Q$ applies an aggregate function, possibly with grouping, on an expression $\raexpr$ and is the last step in the query execution. In our language, a top-level query $Q$ takes a database as input and returns a real number for each of its output grouping bins, $Q: \mathcal{D} \rightarrow \mathbb{R}^n$.

The rest of this section details how to extend the inference system of Figure~\ref{fig:base_relation_rules} to support the aforementioned relational operators. Then, after applying the outermost aggregation, we will derive an upper bound on the global sensitivity (formally defined in Section~\ref{sec:sens}), which parametrizes the Laplace mechanism~\cite{diffpbook} used to compute the amount of noise injected into the query answer.

\subsection{Action Inference for Unary Transformations}
\label{sec:inference_trans}

We next extend the plausible deniability actions to relational transformations. Unary transformations (i.e., select, project, grouping) are detailed first, followed by join transformations. 

The rules for unary transformations are summarized in Figure~\ref{fig:unary_transformation_rules}. Action inference is sequential: the action for each transformation is dependent on the actions of the underlying expressions, starting with the base relations. Each expression can be thought of as a tree where the leaf nodes are base relations and the non-leaf nodes are transformations. The action for each node is computed in a bottom-up manner, starting from the leaves and working up to the root. The root contains the action for the entire query $Q$.

\begin{figure}
\framebox{\textbf{Inference Rules for Unary Transformations}}

\begin{mathpar}
\centering
\inferrule*[right=(\textsc{T-Sel1})]
    {\vdash\raexpr: \aadd{a}{}\times\adel{d}{}}
    {\vdash\select_\pred(\raexpr): \aadd{a}{}\times\adel{d}{}}

\inferrule*[right=(\textsc{T-Sel2})]
    {\vdash\raexpr: \arep{k}{\atts} \\\\ \pred\text{ has no }\att\in\atts}
    {\vdash\select_\pred(\raexpr): \arep{k}{\atts}}

\inferrule*[right=(\textsc{T-Sel3})]
    {\vdash\raexpr: \arep{k}{\atts} \\\\ \pred\text{ has some }\att\in\atts}
    {\vdash\select_\pred(\raexpr): \aadd{k}{}\times\adel{k}{}}

\inferrule*[right=(\textsc{T-Prj1})]
    {\vdash\raexpr: \arep{k}{\atts_1} \\\\
    \atts_1\cap\atts_2\neq\emptyset}
    {\vdash\proj_{\atts_2}(\raexpr):\arep{k}{\atts_1\cap\atts_2}}

\inferrule*[right=(\textsc{T-Prj2})]
    {\vdash\raexpr: \arep{k}{\atts_1} \\\\
    \atts_1\cap\atts_2=\emptyset}
    {\vdash\proj_{\atts_2}(\raexpr): \arep{0}{\emptyset}}

\inferrule*[right=(\textsc{T-Prj3})]
    {\vdash\raexpr: \aadd{a}{}\times\adel{d}{}}
    {\vdash\proj_\atts(\raexpr): \aadd{a}{}\times\adel{d}{}}

\inferrule*[right=(\textsc{T-Grp1})]
    {}
    {\vdash\groupby_\atts(\raexpr): \arep{0}{\emptyset}}

\inferrule*[right=(\textsc{T-Grp2})]
    {\vdash\raexpr: \aadd{a}{}\times\adel{d}{}}
    {\vdash\groupby^{\cnt(*)}_\atts(\raexpr): \arep{a+d}{\{\cnt\}}}

\inferrule*[right=(\textsc{T-Grp3})]
    {\vdash\raexpr: \arep{k}{\atts_1} \\\\
    \atts_1\cap\atts_2\neq\emptyset}
    {\vdash\groupby^{\cnt(*)}_{\atts_2}(\raexpr): \arep{2k}{\{\cnt\}}}

\inferrule*[right=(\textsc{T-Grp4})]
    {\vdash\raexpr: \arep{k}{\atts_1} \\\\
    \atts_1\cap\atts_2=\emptyset}
    {\vdash\groupby^{\cnt(*)}_{\atts_2}(\raexpr): \arep{0}{\emptyset}}
\end{mathpar}
\caption{Action inference rules for select $(\select_{\pred})$, project $(\proj_{\atts})$, grouping $(\groupby_{\atts})$, and grouping with count $(\groupby^{\cnt(*)}_{\atts})$.}
\label{fig:unary_transformation_rules}
\end{figure}

\paragraph{Select}
Let $\raexpr$ be an expression and $\select_\pred(\raexpr)$ be a selection with a predicate $\pred$. If $\raexpr$ has action $\aadd{a}{}\times\adel{d}{}$, then $\select_\pred(\raexpr)$ has the same action since selection does not change rows (\textsc{T-Sel1}). If $\raexpr$ has action $\arep{k}{\atts}$, then the action on $\select_\pred(\raexpr)$ depends on if $\pred$ conditions on any $\att\in\atts$. If it does not, then selection will not affect the number of records in the resulting neighboring tables as it always filters the same records in $\raexpr$ and its neighbor. Hence, the action on the result is still $\arep{k}{\atts}$ (\textsc{T-Sel2}). Otherwise, each replaced value of $\att$ may add, delete, or replace a record in the resulting neighboring tables. Hence, the action on $\select_\pred(\raexpr)$ is $\aadd{k}{}\times\adel{k}{}$ since replacement can be modeled as simultaneous deletion and addition (\textsc{T-Sel3}).

\paragraph{Project}
Let $\raexpr$ be an expression and $\project_\atts(\raexpr)$ be a projection of attributes $\atts$. If $\raexpr$ has action $\arep{k}{\atts_1}$, then the action on $\project_{\atts_2}(\raexpr)$ depends on if $\atts_1$ and $\atts_2$ share any attributes. If they do, then only attributes in $\atts_1\cap\atts_2$ can be replaced (\textsc{T-Prj1}). Otherwise, no attributes in $\project_{\atts_2}(\raexpr)$ are replaceable, so the projection is effectively public (\textsc{T-Prj2}). If $\raexpr$ has action $\aadd{a}{}\times\adel{d}{}$, then $\project_\atts(\raexpr)$ has the same action since projection does not change rows (\textsc{T-Prj3}).

\begin{figure*}
    \centering
    \includegraphics[width=.7\linewidth]{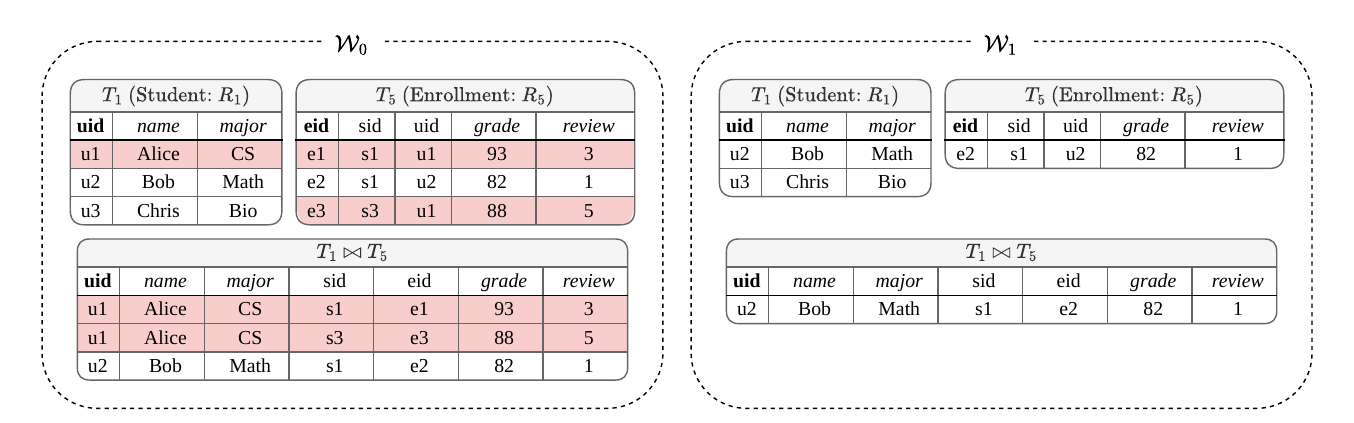}
    \caption{Records in red exist in the hypothetical world $\world_0$, but are deleted in counterfactual world  $\world_1$. In the worst case, at most 2 records are deleted from the join $\tbl_1\join\tbl_5$.}
    \label{fig:join_example}
\end{figure*}

\paragraph{Grouping with Aggregation}

A \texttt{GROUP BY} command in SQL can apply aggregate functions like \texttt{SUM} and \texttt{COUNT} to each group. Providing differential privacy in this case is subtle as group keys may leak information due to the existence/non-existence of a group. For example, Figure~\ref{fig:join_example} shows counterfactual world $\world_1$, where a single student is removed in the true world $\world_0$. Taking the grouping $\groupby_{\{\text{major}\}}(\tbl_1)$ in both worlds results in different set of groups: in $\world_0$, the CS major group exists, whereas in $\world_1$ it does not. 

Following previous work~\cite{FLEX}, \toolname preserves privacy by ensuring that the grouping bins (i.e., the set of values of the grouping attributes) are public. Specifically, \toolname ensures that the set of grouping bins is equal to the domain of the grouping attributes $\atts$. Hence, the result of grouping without aggregation, which simply returns grouping bins, is public (\textsc{T-Grp1}).

The inference system also supports grouping with count, of the form $\groupby^{\cnt(*)}_{\atts}(\raexpr)$. In addition to the public grouping bins, this operator also returns the counts within each bin, akin to a histogram. 
In all cases, the histogram has a $\arep{}{}$ action since the group keys
are public. If $\raexpr$ has action $\aadd{a}{}\times\adel{d}{}$, each added and deleted record may change the count of at most one group. Therefore, the action for $\groupby^{\cnt(*)}_{\atts_2}(\raexpr)$ is $\arep{a+d}{\{\cnt\}}$ (\textsc{T-Grp2}). If $\raexpr$ has action $\arep{k}{\atts_1}$, then the action for $\groupby^{\cnt(*)}_{\atts_2}(\raexpr)$ depends on if $\atts_1$ and $\atts_2$ share any attributes. If they do, then each replaced record may change the count of 2 groups (i.e., add to the count of one group, and reduce the count of another). Therefore, the action is $\arep{2k}{\{\cnt\}}$ (\textsc{T-Grp3}). Otherwise, the replacement cannot affect the count, so the action is $\arep{0}{\emptyset}$ (\textsc{T-Grp4}). 

Note that \textsc{T-Grp3} and \textsc{T-Grp4} also apply to aggregation without grouping when $\atts_2=\emptyset$, a special case when all records are put in one group. For example, $\groupby^{\cnt(*)}_\emptyset(\raexpr)$ returns the size of $S$.

\subsection{Action Inference for Join Transformations}

\begin{figure}
\framebox{\textbf{Inference Rules for Key Joins}}

\begin{mathpar}
\centering
\inferrule*[right=(\textsc{T-Key1})]
    {\vdash\rel: \aadd{0}{}\times\adel{\mown(\ent,\rel)}{} \\
    \rel.\fk_{\raexpr}\fkarrow\raexpr.\pk}
    {\vdash\rel\underset{\fk_{\raexpr}=\pk}{\join}\raexpr: \aadd{0}{}\times\adel{\mown(\ent,\rel)}{}}

\inferrule*[right=(\textsc{T-Key2})]
    {\vdash\rel: \arep{k}{\atts} \\
    \rel.\fk_{\raexpr}\fkarrow\raexpr.\pk \\ \privatt(\raexpr.\pk)}
    {\vdash\rel\underset{\fk_{\raexpr}=\pk}{\join}\raexpr: \arep{k}{\atts\cup\attr(\raexpr)}}

\inferrule*[right=(\textsc{T-Key3})]
    {\vdash\rel: \arep{{k}_1}{\atts_1} \\ \vdash\raexpr: \arep{{k}_2}{\atts_2} \\
    \rel.\fk_{\raexpr}\fkarrow\raexpr.\pk \\ {\neg\texttt{private}(S.PK)}}
    {\vdash\rel\underset{\fk_{\raexpr}=\pk}{\join}\raexpr: \arep{{k}_1}{\atts_1\cup\atts_2}}
\end{mathpar}
\framebox{\textbf{Inference Rules for General Joins}}

\begin{mathpar}
\centering
\inferrule*[right=(\textsc{T-Join1})]
    {\vdash\raexpr_1: \aadd{a_1}{}\times\adel{d_1}{} \\ \vdash\raexpr_2: \aadd{a_2}{}\times\adel{d_2}{}}
    {\vdash\raexpr_1\underset{\att_1=\att_2}{\join}\raexpr_2: \aadd{\substack{a_1\cdot\mmf(\raexpr_2.\att_2)+ \\ a_2\cdot\mmf(\raexpr_1.\att_1)+\\a_1\cdot a_2}}{}{}\times\adel{\substack{d_1\cdot\mmf(\raexpr_2.\att_2)+ \\ d_2\cdot\mmf(\raexpr_1.\att_1)}}{}}

\inferrule*[right=(\textsc{T-Join2})]
    {\vdash\raexpr_1: \aadd{a}{}\times\adel{d}{} \\ \vdash\raexpr_2: \arep{k}{\atts}}
    {\vdash\raexpr_1\underset{\att_1=\att_2}{\join}\raexpr_2: \aadd{\substack{a\cdot\mmf(\raexpr_2.\att_2)+ \\ {k}\cdot\mmf(\raexpr_1.\att_1)+\\a\cdot {k}}}{}{}\times\adel{\substack{d\cdot\mmf(\raexpr_2.\att_2)+ \\ {k}\cdot\mmf(\raexpr_1.\att_1)}}{}}

\inferrule*[right=(\textsc{T-Join3})]
    {\vdash\raexpr_1: \arep{{k}_1}{\atts_1} \quad \vdash\raexpr_2: \arep{{k}_2}{\atts_2} \quad 
    \att_1\notin\atts_1 \quad \att_2\notin\atts_2}
    {\vdash\raexpr_1\underset{\att_1=\att_2}{\join}\raexpr_2: \arep{\substack{{k}_1\cdot\mmf(\raexpr_2.\att_2)+ {k}_2\cdot\mmf(\raexpr_1.\att_1)}}{\atts_1\cup\atts_2}}

\inferrule*[right=(\textsc{T-Join4})]
    {\vdash\raexpr_1: \arep{{k}_1}{\atts_1} \\ \vdash\raexpr_2: \arep{{k}_2}{\atts_2} \\
    \att_1\in\atts_1}
    {\vdash\raexpr_1\underset{\att_1=\att_2}{\join}\raexpr_2: \aadd{\substack{{k}_1\cdot\mmf(\raexpr_2.\att_2)+ \\ {k}_2\cdot\mmf(\raexpr_1.\att_1)+\\{k}_1\cdot {k}_2}}{}{}\times\adel{\substack{{k}_1\cdot\mmf(\raexpr_2.\att_2)+ \\ {k}_2\cdot\mmf(\raexpr_1.\att_1)}}{}}
\end{mathpar}
\caption{Action inference rules for key joins, where the join is on a primary-foreign key pair, and general joins, where the join keys may be any attributes.}
\label{fig:join_rules}
\end{figure}

A join is either a \emph{key join}, where the join is on a primary-foreign key attribute pair, or a \emph{general join}, where the join is on any attributes.

\subsubsection{Key Joins}
A key join has the form of $\rel\underset{\fk_{\raexpr}=\pk}{\join}\raexpr$ which joins on the primary key of an \emph{expression} $\raexpr$ and a matching foreign key $\fk_{\raexpr}$ of a \emph{base relation} $\rel$. In contrast to the general join (when either the foreign key is \emph{not} from a base relation, or when the join is on arbitrary attributes, which we elaborate in Section~\ref{sec:generaljoin}), \toolname leverages the integrity constraint $\rel.\fk_{\raexpr}\fkarrow\raexpr.\pk$ to derive a tighter bound, as shown at the top of Figure~\ref{fig:join_rules}. 

Suppose $\rel$ has action $\aadd{0}{}\times\adel{\mown(\ent,\rel)}{}$. After the join, there are at most $\mown(\ent,\rel)$ records deleted (\textsc{T-Key1}) since due to the foreign key constraint $\rel.\fk_{\raexpr}\fkarrow\raexpr.\pk$, all deleted records from $\raexpr$ match with deleted records from $\rel$.
Since $\rel$ may delete more records due to other foreign keys, the number of total deleted records is upper bounded by the maximum ownership $\mown(\ent,\rel)$.

If instead $\rel$ has action $\arep{k}{\atts}$, then at most $k$ records are replaced. There are two cases: (1) $\raexpr.\pk$ is private or (2) $\raexpr.\pk$ is public. When it is private, each replaced record in $\rel$ may be matched with a different record in $\raexpr$. So, while the size of the join remains the same, the attributes in $\attr(\raexpr)$ are also replaced (\textsc{T-Key2}). When it is public, $\raexpr$ must have a $\arep{}{}$ action. In this case, the size of the join also remains the same, but only records in $\atts_1\cup\atts_2$ may be replaced (\textsc{T-Key3}).

\subsubsection{General Joins}
\label{sec:generaljoin}

The inference rules for general joins are shown at the bottom of Figure~\ref{fig:join_rules}. The maximum frequency upper bound $\mmf(\raexpr.\att)$ is derived for each attribute $\att$ of any expression $\raexpr$. Since these rules are straightforward, they are included in Appendix~\ref{sec:transformation_mmf}.

Suppose $\raexpr_1$ has action $\aadd{a_1}{}\times\adel{d_1}{}$ and $\raexpr_2$ has action $\aadd{a_2}{}\times\adel{d_2}{}$. Each record deleted from $\raexpr_1$ may match at most $\mmf(\raexpr_2.\att_2)$ records in $\raexpr_2$. Conversely, each record deleted from $\raexpr_2$ may match at most $\mmf(\raexpr_1.\att_1)$ records in $\raexpr_1$. Therefore, at most $d_1\cdot\mmf(\raexpr_2.\att_2)+d_2\cdot\mmf(\raexpr_1.\att_1)$ records may be deleted in the result of the join. Similar reasoning applies to the added records, with an additional summed $a_1\cdot a_2$ term since each pair of added records may match (\textsc{T-Join1}). Suppose $\raexpr_2$ instead has action $\arep{k}{\atts}$. Similar reasoning applies, where we conservatively treat $\arep{k}{\atts}$ as $\aadd{k}{}\times\adel{k}{}$
(\textsc{T-Join2}). Note that due to the symmetry of the join operator, \textsc{T-Join2} also applies in the case where the actions on $S_1$ and $S_2$ are swapped.

The remaining cases occur when $\raexpr_1$ has action $\arep{k_1}{\atts_1}$ and $\raexpr_2$ has action $\arep{k_2}{\atts_2}$. If neither join key is replaceable, then the size of the join remains the same. Therefore, the attributes in $\atts_1\cup\atts_2$ of at most ${k}_1\cdot\mmf(\raexpr_2.\att_2)+{k}_2\cdot\mmf(\raexpr_1.\att_1)$ records may be replaced (\textsc{T-Join3}). Otherwise, the size of the join may change, and at most the same number of records may be added and deleted (\textsc{T-Join4}).

\begin{example}[Action Calculation]
    We illustrate that for key joins, the \textsc{T-Key1} rule derives a tighter bound than the \textsc{T-Join1} rule. Consider the university database from Figure~\ref{fig:university} and suppose an upper bound on the maximum frequency (\mmf) of any foreign key is $2$. Let Student ($\rel_1$) be the distinguished entity relation with $\ppolicy_{\rel_1}(\rel_1)=\pdel$ and Enrollment ($\rel_5$) a non-entity relation with $\ppolicy_{\rel_1}(\rel_5)=\pdel$. The inference system first derives $\vdash R_1:\aadd{0}{}\times\adel{1}{}$ by \textsc{E-Del}, then $\vdash R_5:\aadd{0}{}\times\adel{2}{}$ by \textsc{F-Del}. Figure~\ref{fig:join_example} shows a counterfactual world ($\world_1$) where Alice is deleted from $\tbl_1$. Clearly, at most 2 records are deleted from $\tbl_1\join\tbl_5$ as a result. However, there are two ways to derive the action for $\rel_5\underset{\fk_{\rel_1}=\pk}{\join}\rel_1$. First, by \textsc{T-Join1} we have
    \begin{align*}
        \vdash R_5\underset{FK_{R_1}=PK}{\bowtie}R_1:&~\aadd{\substack{a_1\cdot\mmf(\rel_1.\pk)+ \\a_2\cdot\mmf(\rel_5.\fk_{\rel_1})+\\{a_1\cdot a_2}}}{}\times\adel{\substack{d_1\cdot\mmf(\rel_1.\pk)+ \\ d_2\cdot\mmf(\rel_5.\fk_{\rel_1})}}{} \\
        :&~\aadd{0\cdot1+0\cdot2+0\cdot0}{}\times\adel{2\cdot1+1\cdot2}{} \\
        :&~\aadd{0}{}\times\adel{4}{}.
    \end{align*}
    This is an overapproximation since $4>2$. Second, by \textsc{T-Key1} we have $\vdash R_5\underset{FK_{R_1}=PK}{\bowtie}R_1:\aadd{0}{}\times\adel{2}{}$, which is equal to the true stability. We highlight that this discrepancy may be exacerbated with consecutive joins.
\end{example}

\subsection{Global Sensitivity}\label{sec:sens}

The methods presented thus far infer an upper bound on the number of records that might be changed in a neighboring database when a single entity in the distinguished entity relation $\ent$ changes. 
The upper bound provides crucial information to derive the sensitivity of an SQL query.

The \emph{global sensitivity} is the maximum change in the query result on any two neighboring databases.

\begin{definition}[Global Sensitivity]
    Let $\schema$ be a schema and $\ppolicy_{\schema,\ent}$ be a privacy policy. For any query $Q$, the global sensitivity is
    \begin{displaymath}
        \Delta_{\ppolicy_{\schema,\ent}}(Q)=\max_{ (\db,\db')\in\nbrs(\ppolicy_{\schema,\ent})}||Q(\db)-Q(\db'))||_1
    \end{displaymath}
    where $||\cdot||_1$ is the $L_1$ norm\footnote{To be exact, $Q(D)$ is a view with possibly more than one column due to grouping attributes. We write $Q(D)$ here to mean only the vector of aggregated values.}.
\end{definition}

\subsection{Plausible Deniability Actions to Sensitivity}
\label{sec:sens_inference}

As is standard in most of the DP literature, we assume  
each numerical attribute has a bounded range of possible values, which is commonly linked to the attribute's datatype in a database management system. 
We write $\range(\raexpr.\att)=[L,U]$ to specify the lower and upper bounds of possible numeric values for $\raexpr.\att$, where $L\leq U$. In certain cases, the range of an attribute can be updated to a tighter bound, which can decrease the sensitivity of \fsum.
 
The sensitivity of $\cnt$ is dependent on the underlying expression's action. The sensitivity of $\fsum$ is dependent on both the underlying expression's action and the range of the attribute that is being summed. Additionally, \toolname supports the $\avg$ function. The sensitivity of $\avg$ is not computed directly. Instead, a differentially private answer is achieved by the division of a sum query and a counting query to \toolname.

\begin{figure}
\centering
\framebox{\textbf{Sensitivity Rules for Queries}}

\begin{mathpar}
\inferrule*[right=(\textsc{S-Cnt1})]
    {\vdash\raexpr: \arep{{k}}{\atts_1} \\\\
    \atts_1\cap\atts_2\neq\emptyset}
    {\hat\Delta\left(\groupby^{\cnt(*)}_{\atts_2}(\raexpr)\right)=2{k}}

\inferrule*[right=(\textsc{S-Cnt2})]
    {\vdash\raexpr: \arep{{k}}{\atts_1} \\\\
    \atts_1\cap\atts_2=\emptyset}
    {\hat\Delta\left(\groupby^{\cnt(*)}_{\atts_2}(\raexpr)\right)=0}

\inferrule*[right=(\textsc{S-Cnt3})]
    {\vdash\raexpr: \aadd{a}{}\times\adel{d}{}}
    {\hat\Delta\left(\groupby^{\cnt(*)}_\atts(\raexpr)\right)=a+d}

\inferrule*[right=(\textsc{S-Sum1})]
    {\vdash\raexpr: \arep{{k}}{\atts_1} \\\\ 
    \att\notin\atts_1 \\\\ \atts_1\cap\atts_2=\emptyset}
    {\hat\Delta\left(\groupby^{\fsum(\att)}_{\atts_2}(\raexpr)\right)=0}

\inferrule*[right=(\textsc{S-Sum2})]
    {\vdash\raexpr: \arep{{k}}{\atts_1} \\
    \range(\raexpr.\att)=[L,U] \\
    \att\in\atts_1 \\ \atts_1\cap\atts_2=\emptyset}
    {\hat\Delta\left(\groupby^{\fsum(\att)}_{\atts_2}(\raexpr)\right)={k}\cdot|L-U|}

\inferrule*[right=(\textsc{S-Sum3})]
    {\vdash\raexpr: \arep{{k}}{\atts_1} \\
    \range(\raexpr.\att)=[L,U] \\
    \atts_1\cap\atts_2\neq\emptyset}
    {\hat\Delta\left(\groupby^{\fsum(\att)}_{\atts_2}(\raexpr)\right)=2{k}\cdot\max(|L|,|U|)}

\inferrule*[right=(\textsc{S-Sum4})]
    {\vdash\raexpr: \aadd{a}{}\times\adel{d}{} \\
    \range(\raexpr.\att)=[L,U]}
    {\hat\Delta\left(\groupby^{\fsum(\att)}_{\atts_2}(\raexpr)\right)=(a+d)\cdot\max(|L|,|U|)}
\end{mathpar}
\caption{Inference rules to calculate $\hat\Delta(Q)$, an upper bound on the global sensitivity $\Delta(Q)$. The final aggregation function may be immediately preceded by a grouping operation.}
\label{fig:sens_rules}
\end{figure}

Similar to \privatesql~\cite{kotsogiannis2019privatesql}, we derive an upper bound on the global sensitivity. Rules to derive this bound for each aggregation function are shown in Figure~\ref{fig:sens_rules}. Here, $\hat\Delta$ denotes the upper bound on the global sensitivity.

If $\raexpr$ has action $\arep{{k}}{\atts_1}$, the global sensitivity bound for \cnt of $\groupby_{\atts_2}(\raexpr)$ depends on whether $\atts_1$ and $\atts_2$ share any attributes. If they do, then each replacement at worst decrements the count of one group and increments the count of another in the neighboring view. Since each group changes by at most 1, the bound on the global sensitivity is $2{k}$ (\textsc{S-Cnt1}). Otherwise, each replacement makes no change to the group counts of the neighboring view. Therefore, the bound on the global sensitivity is 0 (\textsc{S-Cnt2}). If $\raexpr$ has action $\aadd{a}{}\times\adel{d}{}$, then at most $a$ records are added and $d$ records are deleted in the neighboring view. Therefore, the global sensitivity is at most $a+d$ (\textsc{S-Cnt3}).

If $\raexpr$ has action $\arep{{k}}{\atts_1}$, the global sensitivity bound for $\fsum(\att)$ of $\groupby_{\atts_2}(\raexpr)$ depends on whether $\atts_1$ and $\atts_2$ share any attributes. If they do not, and $\att\notin\atts_1$, then the sum does not change and the bound is 0 (\textsc{S-Sum1}). When $\att\in\atts_1$, then each replaced value of $\att$ at worst changes the sum by $|L-U|$, so the bound is ${k}\cdot|L-U|$ (\textsc{S-Sum2}). If they do share attributes, then replacing a group changes the sum of a group by at most $\max(|L|,|U|)$. Since two groups change, the bound is $2{k}\cdot\max(|L|,|U|)$ (\textsc{S-Sum3}). If $\raexpr$ has action $\aadd{a}{}\times\adel{d}{}$, then at worst each added or deleted record has a value of $\max(|L|,|U|)$. Therefore, the bound is $(a+d)\cdot\max(|L|,|U|)$ (\textsc{S-Sum4}).

\section{Soundness and Privacy}\label{sec:soundness}

The next step is to establish the soundness of the inference system, where the formal proof is included in Appendix~\ref{sec:soundness_proofs}. To formalize soundness, we first define an \emph{action-induced neighboring relation}.

\subsection{Action-Induced Neighbors}

A plausible deniability action $\tau$ can be thought of as being applied to a view $\view$ by changing  at most the number of specified records in $\tau$. This results in the set of action-induced neighbors.

\begin{definition}[Action-Induced Neighbor]\label{def:action-induced}
    Let $\view\subset\dom(\raexpr)$ be an instance of expression $\raexpr$ and $\tau$ be an action. The set of action-induced neighbors of $\view$ by $\tau$ is the set of views $\action(\view,\tau)$ such that $\forall\view'\in\action(\view,\tau)$:
    \begin{itemize}
        \item If $\tau=\aadd{a}{}\times\adel{d}{}$, then
        \begin{itemize}
            \item $|\view'\setminus\view|\leq a$ and $|\view\setminus\view'|\leq d$.
        \end{itemize}
        \item If $\tau=\arep{{k}}{\atts}$, then
        \begin{itemize}
            \item $|\view|=|\view'|$ and
            \item $O=\view\setminus\view'$ and $O'=\view'\setminus\view$ such that $|O|=|O'|\leq {k}$ and $\forall\record\in O.~\exists \record'\in O'$ such that $\record$ only differs from $\record'$ by the values of attributes in $\atts$.
        \end{itemize}
    \end{itemize}
\end{definition}

For example, $\action(\view,\aadd{2}{}\times\adel{3}{})$ is the set of all views that add at most 2 records to $\view$ and delete at most 3 records from $\view$. The set of action-induced neighbors of a view $\view$ is used to show that the inference system is sound.

\subsection{Soundness}

We first prove that for any expression $\raexpr$ and any pair of neighboring databases $(\db,\db')$, $\raexpr(\db)$ and $\raexpr(\db')$ are action-induced neighbors.

\begin{theorem}[Action Soundness]\label{thm:action_soundness}
    Let $\schema$ be a schema and $\ppolicy_{\schema,\ent}$ be a privacy policy. For any expression $\raexpr$ and action $\tau$ such that $\vdash S:\tau$,
    \begin{displaymath}
        (\db,\db')\in\nbrs(\ppolicy_{\schema,\ent})\implies\raexpr(\db')\in\action(\raexpr(\db),\tau).
    \end{displaymath}
\end{theorem}


The sensitivity soundness theorem states that $\hat\Delta$ is an upper bound on the global sensitivity $\Delta$. The theorem assumes that the action environment is sound by Theorem~\ref{thm:action_soundness}.

\begin{theorem}[Sensitivity Soundness]\label{thm:sens_soundness}
    Let $\schema$ be a schema and $\ppolicy_{\schema,\ent}$ be a privacy policy. For any query $Q$,
    \begin{equation*}
        \Delta_{\ppolicy_{\schema,\ent}}(\query)\leq\hat\Delta_{\ppolicy_{\schema,\ent}}(\query).
    \end{equation*}
\end{theorem}

The following soundness theorem is a direct consequence of Theorem~\ref{thm:sens_soundness} and the Laplace mechanism.

\begin{theorem}\label{thm:soundness}
    Let $\db$ be a database with schema $\schema$, $\query: \dbs \rightarrow \mathbb{R}^{n}$ a SQL query, and $\ppolicy_{\schema,\ent}$ a privacy policy. By returning $\query(\db)+(\delta_1,\dots, \delta_{n})$, where $\delta_1,\dots,\delta_{n}$ are i.i.d. random variables drawn from $ \operatorname{Lap}\!\left(\frac{\hat\Delta_{\ppolicy_{\schema,\ent}}(Q)}{\epsilon}\right)$, \toolname satisfies $(\ppolicy_{\schema,\ent},\epsilon)$-differential privacy.
\end{theorem}

\section{Case Study}\label{sec:case_study}

We begin by examining the university schema $\schema=(\rels,\constr)$ shown in Figure~\ref{fig:university} as a case study to illustrate the advantages of supporting flexible privacy policies in \toolname. Specifically, we consider a flexible privacy policy, denoted $\ppolicy^1_{\rel_1}$. This policy is motivated by realistic privacy requirements commonly found in university information systems. In the whole case study, we choose Student ($\rel_1$) as the distinguished entity relation $\ent$. As such, we omit the subscript in the notation for clarity hereafter.

For $\ppolicy^1$, the entirety of the Student relation ($\rel_1$) is assumed to be private (i.e., $\pdel$). The number of Faculty ($\rel_2$) is public, though demographic and financial information is private (i.e., $\prep^{\{\text{salary,age}\}}$). The total number of scholarships ($\rel_3$) awarded to students is public, but who they are awarded to and the amount given is private (i.e., $\prep^{\{\text{uid,amount}\}}$). Information in the course catalog, including the Section relation ($\rel_4$), is considered entirely public (i.e., $\ppub$). The size of the Enrollment relation is also public, reflecting the fact that students can observe enrollment counts through the course registration system; however, the grade associated with each enrollment is private. 
$\ppolicy^1$ assumes uid to be private, protecting the association between students and their reviews while keeping the review scores themselves public (i.e., $\prep^{\{\text{uid,grade}\}}$).

As no existing DP SQL system can support the flexible privacy policy as stated above, we consider two baseline policies $\ppolicy^U$ and $\ppolicy^B$, that best match the capabilities of prior work, for comparison. $\ppolicy^U$ models pure unbounded neighbors as in \privatesql~\cite{kotsogiannis2019privatesql}, where the amount of noise to add to a query is computed by considering the effects of just dropping all records owned by a student (since prior work does not support consistency with public information). This policy has the $\pdel$ privacy label for every relation $\rel\in\rels$. $\ppolicy^B$ models pure bounded neighbors in which the amount of noise is computed by considering arbitrary alterations to all records owned by a student. This policy has the $\prep^{\attr(\rel)}$ privacy label for every relation $\rel\in\rels$. Next, we study the following research questions based on several representative SQL queries:
\begin{enumerate}
    \item[\textbf{RQ1}] \textbf{Privacy Improvement.} When do existing techniques that can only support baseline policies underprotect the data (i.e., inject \emph{less} noise than necessary), and how does \toolname overcome this limitation?

    \item[\textbf{RQ2}] \textbf{Utility Improvement.} When do existing techniques that can only support baseline policies overprotect the data (i.e., inject \emph{more} noise than necessary), and how does \toolname overcome this limitation?
\end{enumerate}

\subsection{Setup}

For any relation $\rel$, the primary key $\rel.\pk$ is unique by definition ($\mmf(\rel.\pk)=1$). In the case study, we assume that faculty teach at most 3 sections ($\mmf(\rel_4.\text{fid})=3$). Each section has at most 20 student enrollments ($\mmf(\rel_5.\text{sid})=20$). Each student is enrolled in at most 6 courses ($\mmf(\rel_5.\text{uid})=6$) and can receive at most 2 scholarships ($\mmf(\rel_3.\text{uid})=2$). Lastly, a review is an integer rating where $\range(\rel_5.\text{review})=[1,5]$.

The privacy budget $\epsilon$ captures the tradeoff between privacy and utility (Definition~\ref{def:policy_dp}). We set $\epsilon=1$ for each SQL query.

\subsection{University Queries}

We consider the following SQL queries in this case study (the exact queries can be found in Appendix~\ref{sec:case_study_queries}):
\begin{itemize}
    \item[$Q1$] \emph{How many sections with less than 10 students is each student enrolled in?}

    \item[$Q2$] \emph{How many times was each review score given in sections with more than 15 students?}

    \item[$Q3$] \emph{How many total scholarships have been awarded to full-time students (i.e., those enrolled in at least 4 sections)?}
    
    \item[$Q4$] \emph{What is the average review score given to each faculty member teaching NLP?}
\end{itemize}

These queries are representative of realistic workloads (e.g., review scores may be used for faculty promotion) and contain a mix of filtering, projection, join, and grouping operators.

\subsection{Privacy and Utility Comparison}

\begin{figure}
\centering
\includegraphics[width=0.5\textwidth]{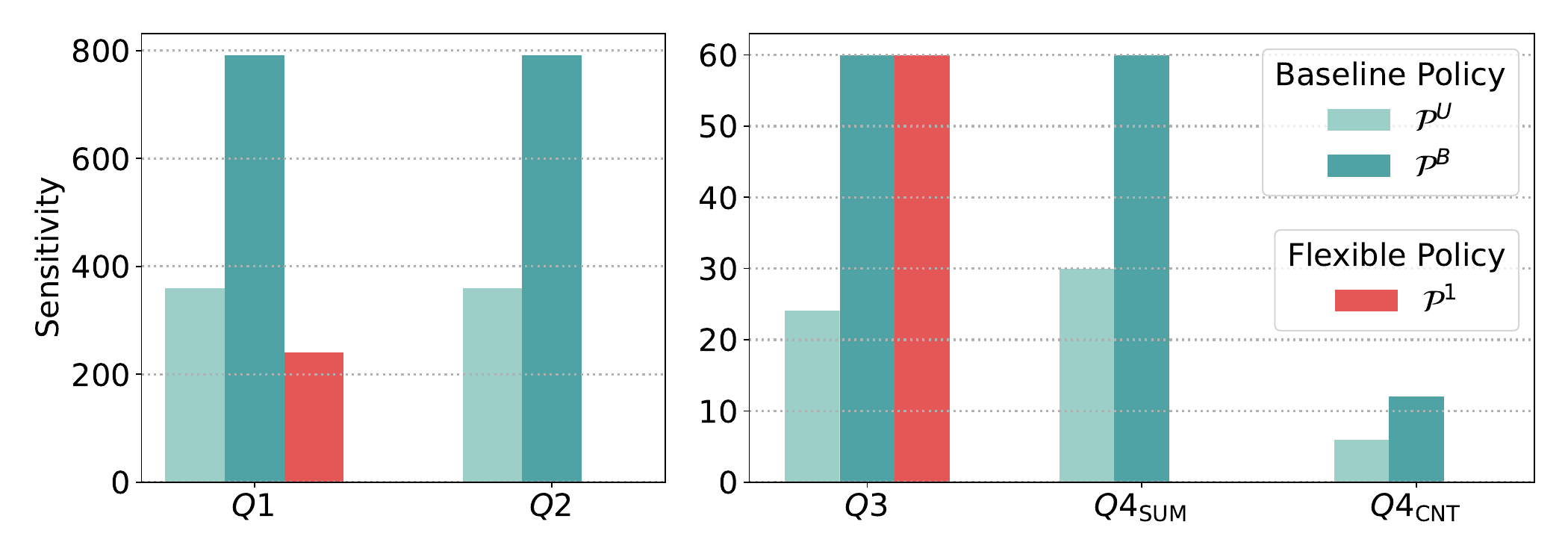}
\caption{Derived sensitivities for each privacy policy and query. Missing values denote zero sensitivity, and hatching indicates underprotection.}
\label{fig:case_study_sens}
\end{figure}

With \toolname, we calculate the sensitivity bound $\hat\Delta$ for each query. We focus on sensitivity comparison when possible as it is directly proportional to the injected noise to query answers.
Figure~\ref{fig:case_study_sens} shows the sensitivity calculated by \toolname for each query and privacy policy combination. 
We next examine each query in detail.

\paragraph{Q1}

This query contains an intermediate grouping with count operation on foreign key $\rel_5.\text{sid}$, followed by a self join on the same key. It also performs a filter on the count followed by a final grouping on $\rel_5.\text{uid}$. The baselines $\ppolicy^U$ and $\ppolicy^B$ both provide overprotection as they treat the grouping attribute $\rel_5.\text{sid}$ as private. They both derive large sensitivity due to the self-join in the query, which uses the general but loose join rule. The difference is proportional to the maximum frequency of the join key $\mmf(\rel_5.\text{sid})=20$. In comparison, $\ppolicy^1$ treats $\rel_5.\text{sid}$ as public, and thus the intermediate grouping with count is also public (see \textsc{T-Grp4}). 

\paragraph{Q2}

This query is similar to $Q1$ except the final grouping on $\rel_5.\text{review}$ instead of $\rel_5.\text{uid}$. Since $\ppolicy^1$ assumes reviews are public, \toolname 
it is able to return the answer to the query without injecting noise (i.e., a sensitivity of zero). This is intuitively correct as releasing statistical enrollment data from $Q2$ does not reveal grade and review attributes.  

\paragraph{Q3}

This simple query illustrates why pure unbounded neighbors ($\ppolicy^U$) is \emph{insufficient} to provide proper differential privacy. Consider the first step of an evaluation of $Q3$: (1) Counting the number of enrollments for each student with a grouping count on the private attribute $\rel_5.\text{uid}$. Since the number of enrollments is public, replacing the uid of a student will change the count of two groups in a counterfactual world, thus increasing the stability by a factor of 2. This is followed by the second step (2), a subsequent join with $\rel_3$ on private foreign key uid, which increases the stability proportional to $\mmf(\rel_3.\text{uid})$. Since $\ppolicy^U$ assumes at most one entity can be deleted, it under-approximates the stability in step (1); intuitively, replacing the uid of a student is the same as deleting one student record \emph{and adding another student record}, which results in a a factor of 2 rather than 1 in stability calculation. The under-approximation with policy $\ppolicy^U$ results in a failure to properly track the overall stability and leading to the incorrect sensitivity displayed in Figure~\ref{fig:case_study_sens}.

\paragraph{Q4}

This query is the most complex, containing two joins, a filter, and a final grouping.
\toolname supports $Q4$ by answering two group-by aggregate subqueries. Let $Q4_\fsum$ and $Q4_\cnt$ be subqueries identical to $Q4$ up to replacing the $\avg$ aggregation with $\fsum$ and $\cnt$ respectively. Then, answering $Q4$ amounts to the quotient $Q4_{\fsum}/Q4_{\cnt}$ of differentially private answers.

The average review score is public information, so \toolname correctly derives that releasing the query answer does not reveal any private data: the query has a sensitivity of zero. Unsurprisingly, however, both baseline policies output high sensitivities for $Q4_{\fsum}$ since they treat review scores as private. For $Q4_{\cnt}$, they both output high sensitivity since they treat all attributes in Enrollment as private.
Notably, $\ppolicy^1$ has zero sensitivity for both subqueries, and therefore would need to inject \emph{zero noise} into the true query answer. 

\subsection{Key Takeaways}

\begin{enumerate}
    \item[\textbf{RQ1}] Existing techniques that are limited to supporting only the unbounded policy $\ppolicy^U$ may introduce insufficient noise, resulting in underprotection. This occurs because they lack the expressiveness needed to account for public table sizes, as demonstrated by $Q3$. Flexible privacy policies overcome this issue with privacy labels that support a mix of public and private information, including table sizes.

    \item[\textbf{RQ2}] Existing techniques that are limited to supporting baseline policies only (e.g., $\ppolicy^U$, $\ppolicy^B$) may introduce excessive noise, resulting in overprotection. This occurs because they conservatively treat all attributes as private. In comparison, \toolname can differentiate between private and non-private attributes, thereby deriving tight sensitivity bounds on query answers. This improvement is demonstrated by the majority of the SQL queries in our case study.
\end{enumerate}

\section{TPC-H Evaluation}\label{sec:evaluation}

We evaluate \toolname on TPC-H~\cite{tpc-h}, an industry standard benchmark. The benchmark dataset consists of eight relations: \text{Region} (\textbf{R}), \text{Nation} (\textbf{N}), \text{Part} (\textbf{P}), \text{Supplier} (\textbf{S}), \text{Partsupp} (\textbf{PS}), \text{Customer} (\textbf{C}), \text{Order} (\textbf{O}), and \text{Lineitem} (\textbf{L}). These relations comprise a complex business environment (e.g., an online shopping platform), whose schema as a data ownership graph is shown in Figure~\ref{fig:tpch_schema}.

\begin{figure}
    \centering
    \includegraphics[width=\linewidth]{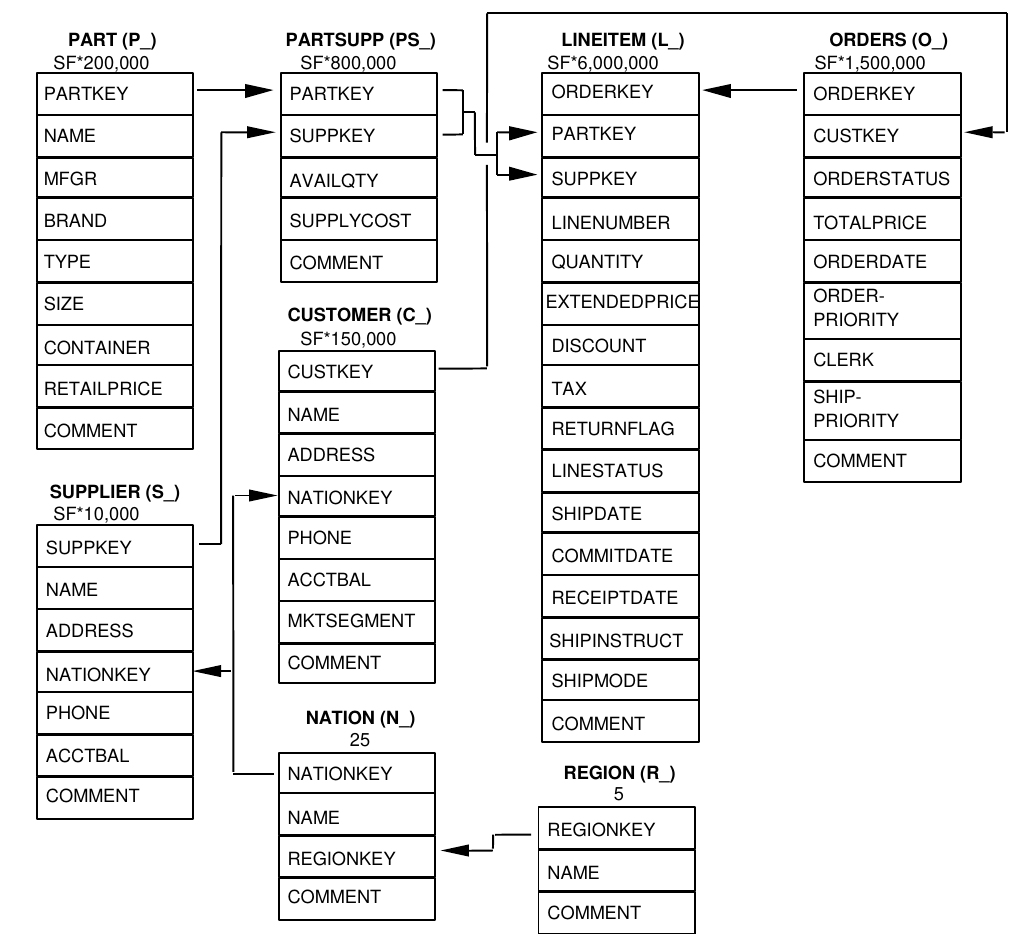}
    \caption{The TPC-H Schema. The arrows point in the direction of data ownership. The number below each table name is the cardinality output by the TPC-H data generation tool (possibly multiplied by a scale factor SF) for that table.}
    \label{fig:tpch_schema}
\end{figure}

We consider two flexible privacy policies; they choose $\textbf{C}$ and $\textbf{S}$ as the distinguished entity relation $\ent$, and are denoted as $\ppolicy_{\textbf{C}}$ and $\ppolicy_{\textbf{S}}$ respectively. In both policies, geographic data about nations (\textbf{N}) and regions (\textbf{R}) is publicly available, as well as part (\textbf{P}) information (i.e., $\ppub$). The number of suppliers (\textbf{S}) using the platform are known, but not their financial or demographic information (i.e., $\prep^{\attr(\textbf{S})}$). The number of parts each supplier sells (\textbf{PS}) is also public (i.e., $\prep^{\attr(\textbf{PS})}$). Customer (\textbf{C}) information, including the number of customers, is highly sensitive and entirely private (i.e., $\pdel$). We assume the platform releases an aggregate number of orders (\textbf{O}) placed by all customers (i.e., $\prep^{\attr(\textbf{O})}$). However, lineitem (\textbf{L}) data, which contains information about each item in an order, is entirely private (i.e., $\pdel$). 

For comparison, we consider two unbounded baseline policies (i.e., $\ppolicy^U_{\textbf{C}}$ and $\ppolicy^U_{\textbf{S}}$) and Tumult Analytics~\cite{tumultanalyticssoftware}, a Python package that builds upon the design principles of \privatesql. Tumult Analytics supports two privacy policies, which we call $\text{Tumult}_{\textbf{C}}$ and $\text{Tumult}_{\textbf{S}}$, that protect the existence of customers and suppliers respectively.

\subsection{Setup}

We generate a database using the TPC-H data generation tool with the default scale factor of 1. This results in 8.7M total records in the database. To compute the upper bounds on maximum frequency, we run the query \texttt{SELECT COUNT(key) AS count FROM table GROUP BY key ORDER BY count DESC LIMIT 1} for each foreign key. 
Since TPC-H queries only perform key joins, upper bounds for non-key attributes are not required.
We use a standard privacy budget of $\epsilon=1$ for each query.

\subsection{Comparison with Tumult Analytics}

\begin{figure}
    \centering
    \includegraphics[width=\linewidth]{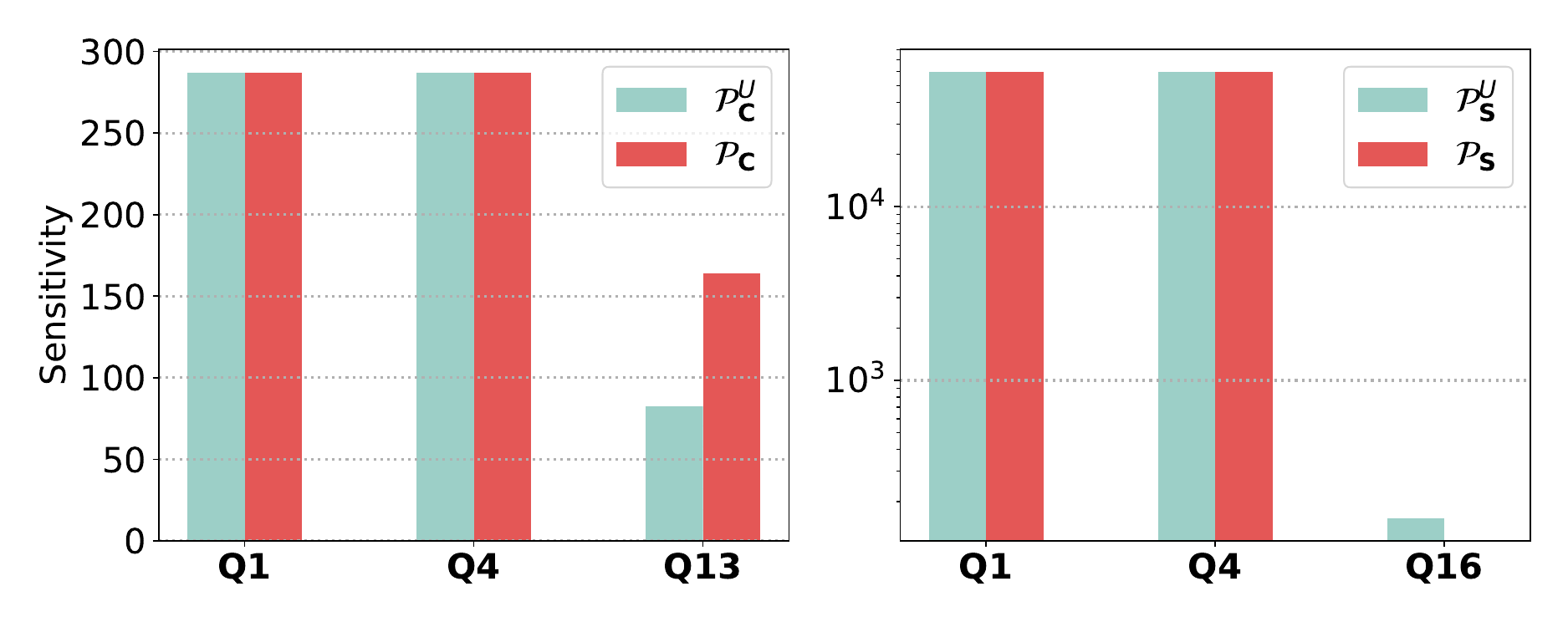}
    \caption{Sensitivities for TPC-H queries. Hatching indicates that the policy underprotects the data.}
    \label{fig:tpch_sens}
\end{figure}

\begin{figure}
    \centering
    \includegraphics[width=\linewidth]{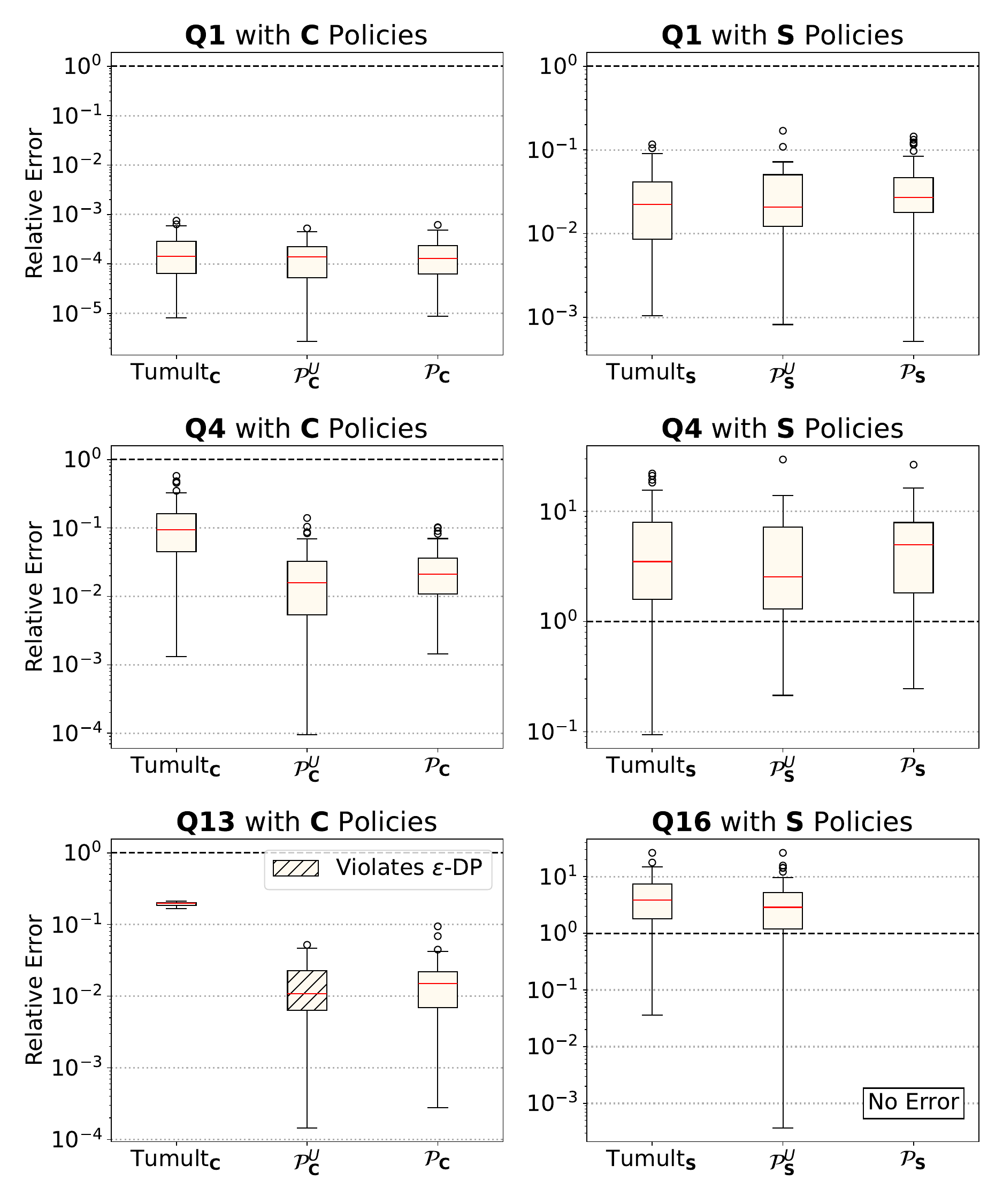}
    \caption{Relative error of \toolname and Tumult Analytics for TPC-H queries. {\textbf{C}} and {\textbf{R}} policies treat Customer and Supplier relations as the entity relations respectively. The dashed line denotes the error for a system that always outputs 0.}
    \label{fig:tpch_error}
\end{figure}

Following existing work~\cite{FLEX,kotsogiannis2019privatesql}, we evaluate on the TPC-H queries that have count aggregations: \textbf{Q1}, \textbf{Q4}, \textbf{Q13}, and \textbf{Q16}. 

\begin{itemize}
    \item [\textbf{Q1}] \emph{Pricing Summary Report}: Outputs the number of lineitems that were shipped before a given date, grouped by return status. It contains no joins.

    \item [\textbf{Q4}] \emph{Order Priority Checking}: Lists the number of orders for which at least one lineitem was received late by its customer. The counts are grouped by the priority of the order. It contains one join on its Orders and Lineitem tables.

    \item [\textbf{Q13}] \emph{Customer Distribution}: Outputs a histogram of how many customers have 1 order, 2 orders, 3 orders, etc. It contains one join on Customer and Orders, as well as an intermediate count aggregation.

    \item [\textbf{Q16}]  \emph{Parts/Supplier Relationship}: Lists the number of suppliers that can satisfy a particular set of part requirements. The counts are grouped by the attributes that make up the requirements. It has two joins: one on Part and Partsupp, followed by another on Supplier.
\end{itemize}

\paragraph{Sensitivity}

Since Tumult Analytics does not disclose the computed sensitivity, we only show the calculated sensitivity for $\ppolicy^U_{\textbf{C}}$, $\ppolicy^U_{\textbf{S}}$, $\ppolicy_{\textbf{C}}$, and $\ppolicy_{\textbf{S}}$ in Figure~\ref{fig:tpch_sens}. All policies exhibit comparable sensitivity for \textbf{Q1}, as they all assume the grouping attributes to be private. A similar observation holds for \textbf{Q4}, where all policies treat lineitems (\textbf{L}) as private. Queries \textbf{Q13} and \textbf{Q16} are more interesting. Since \textbf{Q13} does not access any supplier-owned data, it is evaluated exclusively under the \textbf{C} policies. $\ppolicy^U_{\textbf{C}}$ underprotects the data because the privacy policy incorrectly models the total number of orders on the platform as private. The reasoning is similar to that of $Q3$ in Section~\ref{sec:case_study}.
Conversely, \textbf{Q16} does not involve customer data and is thus evaluated only under the \textbf{S} policies. Since \textbf{Q16} returns public information, \toolname is able to return the query answer without injecting noise (i.e., a sensitivity of zero). In contrast, $\ppolicy^U_{\textbf{S}}$ overprotects data as it assumes the number of suppliers is private. 

\paragraph{Utility}

We now compare the utility of all privacy policies with Tumult Analytics. Given a policy $\ppolicy$ and a query $Q$, we measure utility with the relative error defined as $\texttt{RelError}_{\ppolicy}(Q)=|y-\hat y|/\max(1,y)$ where $y$ is the true answer and $\hat y$ is the noisy answer output by a DP SQL system. For each query $Q$ and each privacy policy $\ppolicy$, we sample the relative error $\texttt{RelError}_{\ppolicy}(Q)$ 50 times. The boxplots are displayed in Figure~\ref{fig:tpch_error}. 

The errors are similar across the board for \textbf{Q1} and \textbf{Q4}, which is consistent with the results from sensitivity analysis. 
$\text{Tumult}_{\textbf{C}}$ overprotects the data for \textbf{Q13} with a median error
of 19.6\%, about 12 times larger than that of $\ppolicy_{\textbf{C}}$. $\text{Tumult}_{\textbf{C}}$ exhibits minimal variability across all samples because deterministic contribution bounding clamps group-by inputs before aggregation, fixing query sensitivity and yielding identical noise distributions across samples. As before, $\ppolicy^U_{\textbf{C}}$ underprotects the data because it incorrectly models the number of platform-wide orders as private. 

Lastly, only $\ppolicy_{\textbf{S}}$ is able to return the true answer to public query \textbf{Q16}. Due to the high maximum ownership $\mown(\textbf{S},\textbf{L})$, the errors output by $\text{Tumult}_{\textbf{S}}$ and $\ppolicy^U_{\textbf{S}}$ are too high for practical use, performing worse than a model that always outputs 0. These results demonstrate that incorrect privacy modeling (e.g., incorrectly assuming the privacy levels of table sizes and/or attributes) can significantly affect noise calculations by resulting in adding too much or too little noise depending on the query.

\section{Conclusion and Future Work}

We introduce \toolname, a flexible DP SQL framework that enables data administrators to meet a mix of complicated privacy needs. Paired with a low-level plausible-deniability-action framework, \toolname reasons soundly about the stability of relational algebra underlying SQL queries and computes a sufficient amount of noise to add to query answers.

For future work, we plan to investigate a high-level declarative policy, such as one that can directly specify privacy requirements for complex views of a database (e.g., protect all grade/review interactions between any faculty/student pair). An extension to \toolname would ingest the specification and produce appropriate privacy policies. A second direction is improving system utility with utility-guided query rewriting. While our work avoids semantics-altering rewrites, some techniques, such as truncation and dropping records, can avoid undesirably large global sensitivities when used properly.

\section*{Ethical Considerations}

All data used in the evaluation of this paper is synthetically generated, publicly available, and is used for demonstration purposes only. As such, there are no stakeholders who are at risk due to the public dissemination of the results produced by \toolname in this paper. 

In a realistic scenario, data processed by \toolname may contain sensitive information owned by real stakeholders. The results produced by \toolname are $\epsilon$-differentially private. Stakeholders should make an informed decision about whether differential privacy is right for their use case before deciding to publicly release any results generated by \toolname.


\begin{acks}
This work was supported by National Science Foundation awards CNS-2317232 and CNS-2317233. 
\end{acks}

\bibliographystyle{ACM-Reference-Format}
\bibliography{references}

\clearpage

\appendix

\section{Case Study SQL Queries}\label{sec:case_study_queries}

The following SQL queries were used in Section~\ref{sec:case_study}.

\begin{enumerate}
  \item[$Q1$] \texttt{SELECT COUNT(*) FROM (SELECT e1.sid, COUNT(*) AS enroll\_count
        FROM enrollment e1
        GROUP BY e1.sid)
  INNER JOIN enrollment e2 ON e1.sid = e2.sid
  WHERE enroll\_count < 10
  GROUP BY e2.uid}
  \item[$Q2$] \texttt{SELECT COUNT(*) FROM (SELECT e1.sid, COUNT(*) AS enroll\_count
        FROM enrollment e1
        GROUP BY e1.sid)
  INNER JOIN enrollment e2 ON e1.sid = e2.sid
  WHERE enroll\_count > 15
  GROUP BY e2.review}
  \item[$Q3$] \texttt{SELECT COUNT(*) FROM (SELECT e.uid, COUNT(*) AS enroll\_count
        FROM enrollment e
        GROUP BY e.uid)
  INNER JOIN scholarship s ON e.uid = s.uid
  WHERE enroll\_count > 4}
  \item[$Q4$] \texttt{SELECT AVG(e.review)
  FROM section s
  INNER JOIN enrollment e ON s.sid = e.sid
  GROUP BY s.fid}
\end{enumerate}
\toolname separates $Q4$ into the following subqueries for sensitivity calculation.
\begin{enumerate}
  \item[$Q4_\fsum$] \texttt{SELECT SUM(e.review)
  FROM section s
  INNER JOIN enrollment e ON s.sid = e.sid
  GROUP BY s.fid}
  \item[$Q4_\cnt$] \texttt{SELECT COUNT(e.review)
  FROM section s
  INNER JOIN enrollment e ON s.sid = e.sid
  GROUP BY s.fid}
\end{enumerate}

\section{Maximum Frequency Upper Bound for Transformations}\label{sec:transformation_mmf}

Recall the maximum frequency upper bound for each attribute in each base relation $\rel\in\rels$ (Section~\ref{sec:mmf}). Calculating the plausible deniability action for $\raexpr_1\underset{\att_1=\att_2}{\join}\raexpr_2$ in the general case requires the maximum frequencies of both $\att_1$ and $\att_2$. Figure~\ref{fig:mf_update} presents the rules for soundly calculating the maximum frequency upper bound for any view produced by a transformation. The following lemma states that \mmf is indeed an upper bound on the maximum frequency for any attribute.

\begin{figure}
\begin{align*}
\mmf(\select_{\pred}(\raexpr).\att)&=\mmf(\raexpr.\att) \\
\mmf(\project_{\atts}(\raexpr).\att)&=\mmf(\raexpr.\att) \\
\mmf(\groupby_{\atts}(\raexpr).\att)&=\mmf(\raexpr.\att) \\
\mmf(\groupby^{\cnt(*)}_{\atts}(\raexpr).\att)&=
\begin{cases}
    \inf & \text{if $\att=\cnt$} \\
    \mmf(\raexpr.\att) & \text{otherwise}
\end{cases} \\
\mmf((\raexpr_1\underset{\att_1=\att_2}{\join}\raexpr_2).\att)&=
\begin{cases}
    \mmf(\raexpr_1.\att)\cdot\mmf(\raexpr_2.\att_2) & \text{if $\att\in\attr(\raexpr_1)$} \\
    \mmf(\raexpr_2.\att)\cdot\mmf(\raexpr_1.\att_1) & \text{otherwise}
\end{cases}
\end{align*}
\caption{Rules for maximum frequency upper bound.}
\label{fig:mf_update}
\end{figure}

\begin{lemma}[Correctness of \mmf]\label{lemma:mf}
    Let $\schema=(\rels,\constr)$ be a schema and $\db\in\dom(\schema)$ any database. For all relations $\rel\in\rels$ and corresponding instances $\tbl\in\db$, assume $\mf(\tbl.\att)\leq\mmf(\rel.\att)$ for all $\att\in\attr(\rel)$. Then,
        $\mf(\raexpr(\db).\att)\leq\mmf(\raexpr.\att).$
    That is, the frequency of the most frequent value of $\raexpr(\db).\att$ is at most $\mmf(\raexpr.\att)$.
\end{lemma}

\section{Soundness Proofs}\label{sec:soundness_proofs}

We prove the main results of Section~\ref{sec:soundness}. 

\subsection{Supporting Lemmas}

We first prove two supporting lemmas. The first is Lemma~\ref{lemma:maximum_ownership} (Correctness of Maximum Ownership):

\begin{proof}
By strong induction on foreign keys.

\emph{Base case}: Each record only owns itself in its own relation. Therefore, for all $\record_1\in\tbl_1$,
\begin{align*}
    |\own(\tbl_1,\tbl_1,\record_1)|=1\leq\mown(\rel_1,\rel_1)=1.
\end{align*}

\emph{Inductive hypothesis}: For all $\rel'$ with associated table $\tbl'\subset\dom(\rel')$ such that $\rel_2\fkarrow\rel'$, assume for all $\record_1\in\tbl_1$
\begin{align*}
    |\own(\tbl_1,\tbl',\record_1)|\leq\mown(\rel_1,\rel').
\end{align*}

\emph{Inductive step}: For each $\rel'$ such that $\rel_2\fkarrow\rel'$, $|\own(\tbl_1,\tbl',\record_1)|\leq\mown(\rel_1,\rel')$ by the inductive hypothesis. For each $\rel_2.\fk_i\fkarrow\rel'.\pk$, at most $\mown(\rel_1,\rel')\cdot\mmf(\rel_2.\fk_i)$ records are owned by $\record_1$. Therefore, for all $\record_1\in\tbl_1$,
\begin{align*}
        |\own(\tbl_1,\tbl_2,\record_1)|&\leq\sum\limits_{\rel_2.\fk_i\fkarrow\rel'.\pk} \mmf(\rel_2.\fk_i)\cdot\mown(\rel_1,\rel') \\
        &=\mown(\rel_1,\rel_2).
    \end{align*}
\end{proof}

\noindent The second is Lemma~\ref{lemma:mf} (Correctness of \mmf):

\begin{proof}\label{proof:mf}
    By induction on the structure of $\raexpr$.

    \textbf{Case} $\rel$. By assumption.

    \textbf{Case} $\select_{\pred}(\raexpr)$.
    Selection does not change any records. Therefore, the result follows from the inductive hypothesis.

    \textbf{Case} $\proj_{\atts}(\raexpr)$.
    Projection does not change any records. Therefore the result follows from the inductive hypothesis.

    \textbf{Case} $\groupby_{\atts}(\raexpr)$.
    Grouping does not add any records. Therefore, the result follows from the inductive hypothesis.

    \textbf{Case} $\groupby^{\{\cnt(*)\}}_{\atts}(\raexpr)$.
    This transformation adds a \cnt attribute for each group. The maximum frequency of \cnt is unbounded since it is at most the number of groups, which is unbounded. For the grouping attributes, the result follows from the inductive hypothesis since grouping does not add records. 

    \textbf{Case} $\raexpr_1\underset{\att_1=\att_2}{\join}\raexpr_2$.
    Suppose $\att\in\attr(\raexpr_1)$. By the inductive hypothesis, $\mmf(\raexpr_1.\att_1),\mmf(\raexpr_1.\att)$, and $\mmf(\raexpr_2.\att_2)$ are upper bounds on the maximum frequencies of $\raexpr_1.\att_1,\raexpr_1.\att$, and $\raexpr_2.\att_2$ respectfully. In the worst case, every record with the most frequent value of $\raexpr_1.\att$ matches with every record with the most frequent value of $\raexpr_2.\att_2$. This is at most $\mmf(\raexpr_1.\att)\cdot\mmf(\raexpr_2.\att_2)$ records, which is precisely the definition. The case where $\att\in\attr(\raexpr_2)$ is symmetric.
\end{proof}

\subsection{Soundness Theorems}

We now prove the two soundness theorems. The first is Theorem~\ref{thm:action_soundness} (Action Soundness):

\begin{proof}
    Assume $(\db,\db')\in\nbrs(\ppolicy_{\schema,\ent})$. We proceed by induction on the base relation inference rules in Figure~\ref{fig:base_relation_rules}: \\

    \textbf{Case} \textsc{E-Del}.
    By Definition~\ref{def:neighbor}, $\tbl_\ent$ and $\tbl'_\ent$ differ by exactly one record and $\tbl'_\ent\setminus\tbl_\ent=\emptyset$. Therefore, $\tbl_\ent'\in\action(\tbl_\ent,\aadd{0}{}\times\adel{1}{})$.

    \textbf{Case} \textsc{E-Pub}.
    By Definition~\ref{def:neighbor}, $\tbl_\ent=\tbl'_\ent$. Therefore, $\tbl_\ent'\in\action(\tbl_\ent,\arep{0}{\emptyset})$.

    \textbf{Case} \textsc{E-Rep}.
    By Definition~\ref{def:neighbor}, $\tbl_\ent$ and $\tbl'_\ent$ are the same size and differ by the values of $\atts$ in exactly one record. Therefore, $\tbl_\ent'\in\action(\tbl_\ent,\arep{1}{\atts})$.

    \textbf{Case} \textsc{F-Pub}.
    By Definition~\ref{def:neighbor}, $\tbl=\tbl'$. Therefore, $\tbl'\in\action(\tbl,\arep{0}{\emptyset})$.

    \textbf{Case} \textsc{F-Del}.
    By Definition~\ref{def:neighbor}, $\tbl'$ deletes $\own(\tbl_\ent,\tbl,e)\subseteq\tbl$. By Lemma~\ref{lemma:maximum_ownership}, this is at most $\mown(\ent,\rel)$ records. Therefore, $\tbl'\in\action(\tbl,\aadd{0}{}\times\adel{\mown(\ent,\rel)}{})$.

    \textbf{Case} \textsc{F-Rep}.
    By Definition~\ref{def:neighbor}, $\tbl'$ replaces the values in $\atts$ of $\own(\tbl_\ent,\tbl,e)\subseteq\tbl$. By Lemma~\ref{lemma:maximum_ownership}, this is at most $\mown(\ent,\rel)$ records. Since $\db'$ satisfies $\constr$, $\atts$ must contain all foreign keys that map to a private primary key. Therefore, $\tbl'\in\action(\tbl,\arep{\mown(\ent,\rel)}{\atts})$. \\

    \noindent Continuing with the unary transformation inference rules in Figure~\ref{fig:unary_transformation_rules}: \\

    \textbf{Case} \textsc{T-Sel1}.
    By the inductive hypothesis, $\vdash\raexpr:\aadd{a}{}\times\adel{d}{}$. In the worst case, added and deleted records are not filtered. Therefore, $\select_\pred(\raexpr(\db'))\in\action(\select_\pred(\raexpr(\db)),\aadd{a}{}\times\adel{d}{})$.

    \textbf{Case} \textsc{T-Sel2}.
    By the inductive hypothesis, $\vdash\raexpr:\arep{k}{\atts}$. By assumption, the predicate $\pred$ does not contain any attribute in $\atts$, so replacements do not affect filtering. Therefore, $\select_\pred(\raexpr(\db'))\in\action(\select_\pred(\raexpr(\db)),\arep{k}{\atts})$.

    \textbf{Case} \textsc{T-Sel3}.
    By the inductive hypothesis, $\vdash\raexpr:\arep{k}{\atts}$. By assumption, the predicate $\pred$ contains some attribute in $\atts$. Each replaced record results in an added, deleted, or replaced record after filtering. Therefore, $\select_\pred(\raexpr(\db'))\in\action(\select_\pred(\raexpr(\db)),\aadd{k}{}\times\adel{k}{})$.

    \textbf{Case} \textsc{T-Prj1}.
    By the inductive hypothesis, $\vdash\raexpr:\arep{k}{\atts_1}$. Projection does not change the number of records and, by assumption, $\atts_1\cap\atts_2\neq\emptyset$. Therefore, $\proj_{\atts_2}(\raexpr(\db'))\in\action(\proj_{\atts_2}(\raexpr(\db)),\arep{k}{\atts_1\cap\atts_2})$.

    \textbf{Case} \textsc{T-Prj2}.
    By the inductive hypothesis, $\vdash\raexpr:\arep{k}{\atts_1}$. By assumption, $\atts_1\cap\atts_2=\emptyset$, so the projection is completely public. Therefore, $\proj_{\atts_2}(\raexpr(\db'))\in \action(\proj_{\atts_2}(\raexpr(\db)),\arep{0}{\emptyset})$.

    \textbf{Case} \textsc{T-Prj3}.
    By the inductive hypothesis, $\vdash\raexpr:\aadd{a}{}\times\adel{d}{}$. Again, projection does not change the number of records. Therefore, $\proj_{\atts_1}(\raexpr(\db'))\in\action(\proj_{\atts_1}(\raexpr(\db)),\aadd{a}{}\times\adel{d}{})$.

    \textbf{Case} \textsc{T-Grp1}.
    By assumption, the grouping bins are public (so the size is constant). Grouping without aggregation simply returns the set of public grouping bins. Therefore, $\groupby_\atts(\raexpr(\db'))\in\action(\groupby_\atts(\raexpr(\db)),\arep{0}{\emptyset})$.

    \textbf{Case} \textsc{T-Grp2}.
    By the inductive hypothesis, $\vdash\raexpr:\aadd{a}{}\times\adel{d}{}$. By assumption, the grouping bins are public (so the size is constant). Each added record may increment a group's count, and each deleted record may decrement a group's count. Therefore, $\groupby^{\{\cnt\}}_\atts(\raexpr(\db'))\in\action(\groupby^{\{\cnt\}}_\atts(\raexpr(\db)),\arep{a+d}{\{\cnt\}})$.

    \textbf{Case} \textsc{T-Grp3}.
    By the inductive hypothesis, $\vdash\raexpr:\arep{k}{\atts_1}$. By assumption, the grouping bins are public (so the size is constant), and grouping attributes can be replaced. Each replaced record may increment one group's count and decrement another group's count. Therefore, $\groupby^{\{\cnt\}}_{\atts_2}(\raexpr(\db'))\in\action(\groupby^{\{\cnt\}}_{\atts_2}(\raexpr(\db)),\arep{2k}{\{\cnt\}})$.

    \textbf{Case} \textsc{T-Grp4}.
    By the inductive hypothesis, $\vdash\raexpr:\arep{k}{\atts_1}$. By assumption, the grouping bins are public (so the size is constant), and grouping attributes cannot be replaced, so the output is public. Therefore, $\groupby^{\{\cnt\}}_{\atts_2}(\raexpr(\db'))\in\action(\groupby^{\{\cnt\}}_{\atts_2}(\raexpr(\db)),\arep{0}{\emptyset})$. \\

    \noindent Finally with the join transformation inference rules in Figure~\ref{fig:join_rules}: \\

    \textbf{Case} \textsc{T-Key1}.
    By the inductive hypothesis, $\vdash\rel:\aadd{0}{}\times\adel{d}{}$. By assumption, $\rel.\fk_{\raexpr}\fkarrow\raexpr.\pk$. By Definition~\ref{def:neighbor}, all changed records from $\raexpr$ match with deleted record from $\rel$. Since $\rel$ may delete more records due to other foreign keys, the number of total deleted records is upper bounded by the maximum ownership $\mown(\ent,\rel)$. So, at most $\mown(\ent,\rel)$ records are deleted in $(\rel\underset{\fk_{\raexpr}=\pk}{\join}\raexpr)(\db')$. Therefore, 
    \begin{displaymath}
        (\rel\underset{\fk_{\raexpr}=\pk}{\join}\raexpr)(\db')\in\action((\rel\underset{\fk_{\raexpr}=\pk}{\join}\raexpr)(\db),\aadd{0}{}\times\adel{d}{}).
    \end{displaymath}

    \textbf{Case} \textsc{T-Key2}.
    By the inductive hypothesis, $\vdash\rel:\arep{k}{\atts}$. By assumption, $\rel.\fk_{\raexpr}\fkarrow\raexpr.\pk$ and $\privatt(\raexpr.\pk)$. By Definition~\ref{def:neighbor}, all changed records from $\raexpr$ match with replaced records from $\rel$. So, at most $k$ records are replaced in $(\rel\underset{\att_1=\att_2}{\join}\raexpr)(\db')$. Each replaced record may replace attributes in $\atts$ as well as $\attr(\raexpr_2)$ due to the replacement of $\fk_{\raexpr}$. Therefore, 
    \begin{displaymath}
        (\rel\underset{\fk_{\raexpr}=\pk}{\join}\raexpr)(\db')\in\action((\rel\underset{\fk_{\raexpr}=\pk}{\join}\raexpr)(\db),\arep{k}{\atts\cup\attr(\raexpr)}).
    \end{displaymath}

    \textbf{Case} \textsc{T-Key3}.
    By the inductive hypothesis, $\vdash\rel:\arep{k_1}{\atts_1}$ and $\vdash\raexpr:\arep{k_2}{\atts_2}$. By assumption, $\rel.\fk_{\raexpr}\fkarrow\raexpr.\pk$. By Definition~\ref{def:neighbor}, all changed records from $\raexpr$ match with replaced records from $\rel$. The size of the join does not change since the join key $\raexpr.\pk$ is public, so at most $k_1$ records are replaced in $(\rel\underset{\fk_{\raexpr}=\pk}{\join}\raexpr)(\db')$. Therefore, 
    \begin{displaymath}
        (\rel\underset{\fk_{\raexpr}=\pk}{\join}\raexpr)(\db')\in\action((\rel\underset{\fk_{\raexpr}=\pk}{\join}\raexpr)(\db),\arep{k_1}{\atts_1\cup\atts_2}).
    \end{displaymath}

    \textbf{Case} \textsc{T-Join1}.
    By the inductive hypothesis, $\vdash\raexpr_1:\aadd{a_1}{}\times\adel{d_1}{}$ and $\vdash\raexpr_2:\aadd{a_2}{}\times\adel{d_2}{}$. There are 5 sources of changed records: (1) A record from $\raexpr_1(\db)$ may match a record added to $\raexpr_2(\db')$, (2) a record from $\raexpr_1(\db)$ may match a record deleted from $\raexpr_2(\db')$, (3) a record from $\raexpr_2(\db)$ may match a record added to $\raexpr_1(\db')$, (4) a record from $\raexpr_2(\db)$ may match a record deleted from $\raexpr_1(\db')$, and (5) a record added to $\raexpr_1(\db')$ may match a record added to $\raexpr_2(\db')$.
    
    Considering each source of changed records: (1) In the worst case, each record added to $\raexpr_2(\db')$ matches with the most popular join key in $\raexpr_1(\db)$, which is at most $\mmf(\raexpr_1.\att_1)$ (Lemma \ref{lemma:mf}). So, at most $a_2\cdot\mmf(\raexpr_1.\att_1)$ records are added. (2) By similar reasoning, at most $d_2\cdot\mmf(\raexpr_1.\att_1)$ records are deleted. (3) By symmetric reasoning, at most $a_1\cdot\mmf(\raexpr_2.\att_2)$ additional records are added and (4) at most $d_1\cdot\mmf(\raexpr_2.\att_2)$ additional records are deleted. (5) At most $a_1\cdot a_2$ additional records are added since each pair of added records may match. The total number of added and deleted records is exactly the action:
    \begin{align*}
        \aadd{\substack{a_1\cdot\mmf(\raexpr_2.\att_2)+ \\ a_2\cdot\mmf(\raexpr_1.\att_1)+ \\ a_1\cdot a_2}}{}{}\times\adel{\substack{d_1\cdot\mmf(\raexpr_2.\att_2)+ \\ d_2\cdot\mmf(\raexpr_1.\att_1)}}{}=\tau
    \end{align*}
    Therefore, $(\raexpr_1\underset{\att_1=\att_2}{\join}\raexpr_2)(\db')\in\action((\raexpr_1\underset{\att_1=\att_2}{\join}\raexpr_2)(\db),\tau)$.

    \textbf{Case} \textsc{T-Join2}.
    By the inductive hypothesis, $\vdash\raexpr_1:\aadd{a}{}\times\adel{d}{}$ and $\vdash\raexpr_2:\arep{k}{\atts}$. There are 4 sources of changed records: (1) A record from $\raexpr_1(\db)$ may match a record replaced from $\raexpr_2(\db')$, (2) a record from $\raexpr_2(\db)$ may match a record added to $\raexpr_1(\db')$, (3) a record from $\raexpr_2(\db)$ may match a record deleted from $\raexpr_1(\db')$, and (4) a record replaced from $\raexpr_2(\db')$ may match a record added to $\raexpr_1(\db')$. 
    
    Considering each source of changed records: (1) In the worst case, each record from $\raexpr_2(\db')$ after replacement matches with the most popular join key in $\raexpr_1(\db)$ (before and after replacement), which is at most $\mmf(\raexpr_1.\att_1)$ (Lemma \ref{lemma:mf}). So, at most $k\cdot\mmf(\raexpr_1.\att_1)$ records are added and deleted. (2) By similar reasoning to \textsc{T-Join1}, at most $a\cdot\mmf(\raexpr_2.\att_2)$ records are added and (3) at most $d\cdot\mmf(\raexpr_2.\att_2)$ records are deleted. (4) At most $a\cdot k$ additional records are added since each pair of added and replaced records may match. The sum of added and deleted records is exactly the action:
    \begin{align*}
        \aadd{\substack{a\cdot\mmf(\raexpr_2.\att_2)+ \\ k\cdot\mmf(\raexpr_1.\att_1)+ \\ a\cdot k}}{}{}\times\adel{\substack{d\cdot\mmf(\raexpr_2.\att_2)+ \\ k\cdot\mmf(\raexpr_1.\att_1)}}{}=\tau
    \end{align*}
    Therefore, $(\raexpr_1\underset{\att_1=\att_2}{\join}\raexpr_2)(\db')\in\action((\raexpr_1\underset{\att_1=\att_2}{\join}\raexpr_2)(\db),\tau)$.

    \textbf{Case} \textsc{T-Join3}.
    By the inductive hypothesis, $\vdash\raexpr_1:\arep{k_1}{\atts_1}$ and $\vdash\raexpr_2:\arep{k_2}{\atts_2}$. There are 3 sources of changed records: (1) A record from $\raexpr_1(\db)$ may match a record replaced from $\raexpr_2(\db')$, and (2) a record from $\raexpr_2(\db)$ may match a record replaced from $\raexpr_1(\db')$
    
    By assumption, the join keys cannot be replaced. So, the size of the join remains constant. Considering each source of changed records: (1) In the worst case, each replaced record in $\raexpr_2(\db')$ matches with the most popular join key in $\raexpr_1(\db)$, which is at most $\mmf(\raexpr_1.\att_1)$ (Lemma \ref{lemma:mf}). So, at most $k_2\cdot\mmf(\raexpr_1.\att_1)$ records are replaced. (2) By symmetric reasoning, at most $k_1\cdot\mmf(\raexpr_2.\att_2)$ additional records are replaced. The total number of replaced records is exactly the action:
    \begin{align*}
        \arep{\substack{k_1\cdot\mmf(\raexpr_2.\att_2)+ \\ k_2\cdot\mmf(\raexpr_1.\att_1)}}{}=\tau
    \end{align*}
    Therefore, $(\raexpr_1\underset{\att_1=\att_2}{\join}\raexpr_2)(\db')\in\action((\raexpr_1\underset{\att_1=\att_2}{\join}\raexpr_2)(\db),\tau)$.

    \textbf{Case} \textsc{T-Join4}.
    By the inductive hypothesis, $\vdash\raexpr_1:\arep{k_1}{\atts_1}$ and $\vdash\raexpr_2:\arep{k_2}{\atts_2}$. There are 2 sources of changed records: (1) A record from $\raexpr_1(\db)$ may match a record replaced from $\raexpr_2(\db')$, (2) a record from $\raexpr_2(\db)$ may match a record replaced from $\raexpr_1(\db')$, and (3) a record replaced from $\raexpr_1(\db')$ may match a record replaced from $\raexpr_2(\db')$.
    
    By assumption, at least one join key can be replaced. So, the size of the join may change. Considering each source of changed records: In the worst case, each replaced record in $\raexpr_2(\db')$ matches with the most popular join key in $\raexpr_1(\db)$ (before and after replacement), which is at most $\mmf(\raexpr_1.\att_1)$ (Lemma \ref{lemma:mf}). So, at most $k_2\cdot\mmf(\raexpr_1.\att_1)$ records are added and deleted. By symmetric reasoning, at most $k_1\cdot\mmf(\raexpr_2.\att_2)$ additional records are added and deleted. (3) At most $k_1\cdot k_2$ additional records are added since each pair of replaced records may match. The total number of added and deleted records is exactly the action:
    \begin{align*}
        \aadd{\substack{k_1\cdot\mmf(\raexpr_2.\att_2)+ \\ k_2\cdot\mmf(\raexpr_1.\att_1)+ \\ k_1\cdot k_2}}{}\times\adel{\substack{k_1\cdot\mmf(\raexpr_2.\att_2)+ \\ k_2\cdot\mmf(\raexpr_1.\att_1)}}{}=\tau
    \end{align*}
    Therefore, $(\raexpr_1\underset{\att_1=\att_2}{\join}\raexpr_2)(\db')\in\action((\raexpr_1\underset{\att_1=\att_2}{\join}\raexpr_2)(\db),\tau)$.
\end{proof}

\noindent The second is Theorem~\ref{thm:sens_soundness} (Sensitivity Soundness):

\begin{proof}
    By induction on the sensitivity rules in Figure~\ref{fig:sens_rules}.

    \textbf{Case} \textsc{S-Cnt1}.
    By the inductive hypothesis, $\vdash\raexpr:\arep{k}{\atts_1}$. Since $\atts_1\cap\atts_2\neq\emptyset$, the grouping attributes $\atts_2$ may be replaced. In the worst case, for each record $\record$, some $\att\in\atts_2$ is replaced such that $\record$ is in a different group in any neighbor. This decreases the count by 1 in its initial group, and increases the count by 1 in its new group. Therefore, the global sensitivity $\Delta$ is at most $2k$.

    \textbf{Case} \textsc{S-Cnt2}.
    By the inductive hypothesis, $\vdash\raexpr:\arep{k}{\atts_1}$. Since $\atts_1\cap\atts_2=\emptyset$, the grouping attributes $\atts_2$ cannot be replaced. Therefore, the count is unchanged in any neighbor, and the global sensitivity $\Delta$ is at most 0.

    \textbf{Case} \textsc{S-Cnt3}.
    By the inductive hypothesis, $\vdash\raexpr:\aadd{a}{}\times\adel{d}{}$. In the worst case, each added and deleted record is in a different group. Therefore, the global sensitivity $\Delta$ is at most $a+d$.

    \textbf{Case} \textsc{S-Sum1}.
    By the inductive hypothesis, $\vdash\raexpr:\arep{k}{\atts_1}$. By assumption, $\att\notin\atts_1$ and $\atts_1\cap\atts_2=\emptyset$, meaning neither the summed attribute nor the grouping attributes can be changed. Therefore, the global sensitivity $\Delta$ is at most 0.

    \textbf{Case} \textsc{S-Sum2}.
    By the inductive hypothesis, $\vdash\raexpr:\arep{k}{\atts_1}$. By assumption, $\att\in\atts_1$ and $\atts_1\cap\atts_2=\emptyset$, meaning the summed attribute can be changed but the grouping attributes cannot be changed. In the worst case, a value at $\att$ is replaced from $L$ to $U$, or from $U$ to $L$. This change is $|L-U|$. Since at most $k$ replacements may occur, the global sensitivity $\Delta$ is at most $k\cdot|L-U|$.

    \textbf{Case} \textsc{S-Sum3}.
    By the inductive hypothesis, $\vdash\raexpr:\arep{k}{\atts_1}$. By assumption, $\atts_1\cap\att_2\neq\emptyset$, meaning the grouping attributes may be changed.
    In the worst case, changing any grouping attribute will cause one record with the most extreme value to move groups. This causes a difference of $\max(|L|,|U|)$ in two groups. Since there are at most $k$ replaced records, the global sensitivity $\Delta$ is at most $2k\cdot\max(|L|,|U|)$.

    \textbf{Case} \textsc{S-Sum4}.
    By the inductive hypothesis, $\vdash\raexpr:\aadd{a}{}\times\adel{d}{}$. Any two neighbors differ by at most $a+d$ records. Each added and deleted record changes the sum of a group by at most $\max(|L|,|U|)$. Therefore, the global sensitivity $\Delta$ is at most $(a+d)\cdot\max(|L|,|U|)$.
\end{proof}

\end{document}